\documentclass[final,twocolumn]{IEEEtran}

\usepackage{amsmath,epsfig,amssymb,algorithm,amsthm,cite,url,pgfplots,tikz}

\usepackage{algpseudocode}
\usepackage{hyperref}

\newtheorem{theorem}{Theorem}
\newtheorem{lemma}[theorem]{Lemma}
\newtheorem{assumption}{Assumption~A-\kern-0pt}

\newtheorem*{remark*}{Remark}
\DeclareMathOperator{\tr}{tr}

\usepackage{pgfplots}

\newcommand{\asto}{\overset{\rm a.s.}{\longrightarrow}}

\title{Convergence and Fluctuations \\
of  Regularized Tyler Estimators }
\author{Abla Kammoun, Romain Couillet, Fr\'ed\'eric Pascal, Mohamed-Slim Alouini
\thanks{A. Kammoun and M.S. Alouini are with the Computer, Electrical, and Mathematical Sciences and Engineering (CEMSE) Division, KAUST, Thuwal, Makkah Province, Saudi Arabia (e-mail: abla.kammoun@kaust.edu.sa, slim.alouini@kaust.edu.sa)}
\thanks{R. Couillet and F. Pascal are with Laboratoire des Signaux et Syst\`emes (L2S, UMR CNRS 8506) CentraleSup\'elec-CNRS-Universit\'e Paris-Sud, 91192 Gif-sur-Yvette, France (e-mail: romain.couillet@centralesupelec.fr, frederic.pascal@centralesupelec.fr)
}
\thanks{Couillet's work is supported by the ERC MORE EC--120133}
}
\begin{document}

\pgfmathdeclarefunction{gauss}{2}{%
  \pgfmathparse{1/(#2*sqrt(2*pi))*exp(-((x-#1)^2)/(2*#2^2))}%
}
\maketitle
\begin{abstract}
This article studies the behavior of regularized Tyler estimators (RTEs) of scatter matrices. The key advantages of these  estimators are twofold. First, they  guarantee by construction a good conditioning of the estimate and second, being a derivative of  robust Tyler estimators, they inherit their robustness properties, notably their resilience to the presence of outliers. Nevertheless, one  major problem that poses the use of  RTEs in practice  is represented by the question  of setting the regularization parameter $\rho$. While a high value of $\rho$ is likely to push all the eigenvalues away from zero, it comes at the cost of  a larger bias with respect to the population covariance matrix. A deep understanding of the statistics of RTEs is essential to come up with appropriate choices for the regularization parameter. This is not an easy task and might be out of reach, unless one considers asymptotic regimes wherein the number of observations $n$ and/or their size $N$ increase  together. First asymptotic results have  recently been obtained  under the assumption that $N$ and $n$ are large and commensurable. Interestingly, no results concerning the regime of $n$ going to infinity  with $N$ fixed exist, even though the investigation of this assumption has usually predated the analysis of the most difficult  $N$ and $n$ large case. 
This motivates our work. In particular, we prove in the present paper  that the RTEs converge to a deterministic matrix when $n\to\infty$ with $N$ fixed, which is expressed as a function of the theoretical covariance matrix. We also derive the fluctuations of the RTEs around this deterministic matrix and establish  that these fluctuations converge in distribution to a multivariate Gaussian distribution with zero mean and a covariance depending on the population covariance and the parameter $\rho$. 



\end{abstract}
\section{Introduction}
The estimation of covariance matrices is at the heart of many applications in signal processing and wireless communications.
The frequently used estimator is the well-known sample covariance matrix (SCM).  Its popularity owes to its low complexity and  in general to a good understanding of its behavior.
However, the use of the SCM in practice is hurdled by its poor performance when samples contain outliers or have an impulsive nature. This is especially the case of radar detection applications in which the noise is often modeled by heavy-tailed distributions \cite{Ward81,Watts85,Nohara91,Billingsley93}.  
One of the reasons why the SCM performs poorly in such scenarios is that, as opposed to the case of Gaussian observations,  the SCM is not the maximum likelihood estimator (MLE) of the covariance matrix. This is for instance the case of complex elliptical distributions, originally introduced by Kelker \cite{kelker} and widely used in radar applications, for which the MLE takes a strikingly different form. 

In order to achieve better robustness against outliers, a class of covariance estimators termed robust estimators of scatter  have been proposed by Huber, Hampel and Maronna  \cite{huber1964robust,Huber72,Maronna76}, and extended more recently  to the complex case \cite{esa-12,mahot-13,pascal2008covariance}. 
This class of estimators can be viewed as a generalization of MLEs, in that they are derived from the optimization of a meaningful cost function \cite{ollila-tyler,Palomar-14}. Aside from robustness to the presence of outliers, a second feature whose importance should not be underestimated, is the conditioning of the  covariance matrix estimate. This feature becomes all the more central when the quantity of interest coincides with the inverse of the population covariance matrix. In order to guarantee an acceptable conditioning, regularized robust-estimators, which find  their roots in the diagonal loading technique  due to Abramowitch and Carlson  \cite{abramovich-81,carlson-88}, were proposed in \cite{ollila-tyler}. The idea is to force by construction all the eigenvalues of the robust-scatter estimator to be greater than a regularization coefficient $\rho$. 

The most popular regularized estimators that are today receiving increasing interest, are  the regularized Tyler estimators (RTE), which correspond to  regularized versions of the robust Tyler estimator \cite{tyler}.
In addition to achieving the desired robustness property, RTEs present the advantage of being well-suited to scenarios where the number of observations is insufficient or the population covariance matrix is ill-conditioned, while  their non-regularized counterparts are ill-conditioned or even undefined if the number of observations $n$ is less than their sizes $N$. Motivated by these interesting features, several works have recently considered the use of RTEs in radar detection applications \cite{chen-11,Pascal-2013,kammoun-15,ollila-tyler,couillet-kammoun-14}. While existence and uniqueness of the robust-scatter estimator seem to be sufficiently studied \cite{Pascal-2013,ollila-tyler}, the impact of the regularization parameter on the behavior of the RTE has remained  less understood. Answering  this question is essential in order to come up with appropriate designs of the RTE in practice. It poses, however, major technical challenges, mainly because it necessitates a profound analysis of the behavior of the RTE estimator, which is far from being an easy task.
As a matter of fact, the main difficulty towards studying the behavior of the RTE fundamentally lies in its non-linear  relation to the observations, thus rendering the analysis for fixed $n$ and $N$  likely out of reach.
In light of this observation, recent works have considered asymptotic regimes where $n$ and/or $N$ are allowed to grow to infinity. Two regimes can be distinguished:  the regime of fixed $N$ with $n$ growing to infinity and the regime  of $n$ and $N$ growing large simultaneously.
While the former regime, coined the large-$n$ regime, is standard in that it was by far the most considered in the literature, the second one, which we will refer to as large-$n,N$ regime, is very recent and is particularly driven by the recent advances in the spectral analysis of large dimensional random matrices.
Interestingly, contrary to what one would imagine, very  little on the behavior of RTE seems to be known in the standard regime, whereas very recent results regarding the behavior of RTE for the large-$n,N$ regime have recently been obtained in \cite{couillet-kammoun-14,couillet-13}. 
One major advantage of the large-$n,N$ regime is that, although requiring the use of advanced tools from random matrix theory, it often leads to less involved results that let themselves to simple interpretation. This interesting feature fundamentally inheres in the double averaging effect that  leads to more compact results in which only prevailing quantities remain. 
However, when $N$ is not so large, the same averaging effect is no longer valid and thus cannot be leveraged. A priori, assuming that $N$ is fixed entails major changes on the behavior of RTEs that have not thus far been  grasped. Understanding what really happens in the large-$n$ regime, besides its own theoretical interest,  should lead to alternative results that might be more accurate for not so large-$N$ scenarios. 
A second motivation behind working under the large-$n$ regime  is that covariance matrix estimators usually converge in this case  to deterministic matrices, which opens up possibilities for easier handling of the RTE.
Encouraged by these interesting practical and theoretical aspects, we study in this paper the asymptotic {behavior} of the RTE in the large-$n$ regime. In particular, we prove  in section \ref{sec:first_order} that the RTE converges to a deterministic matrix which depends on the theoretical covariance {matrix} and the regularization parameter before presenting its fluctuations around this asymptotic limit in section \ref{sec:second_order}. Numerical results are finally provided in order to support the accuracy of the derived results.

{\bf Notation}. In this paper, the following notations are used. Vectors are defined as column vectors and designated with bold lower case, while matrices are given in bold upper case. The norm notation $\|.\|$ refers to the spectral norm for matrices and Euclidean norm for vectors while the norm $\|.\|_{\rm Fro}$ refers to the {Frobenius} norm of matrices. Notations $(.)^{\mbox{\tiny T}}$  $(.)^*$, $\overline{(.)}$  denotes respectively transpose, Hermitian (i.e. complex conjugate transpose) and pointwise conjugate.
Besides, ${\bf I}_N$ denotes the $N\times N$ identity matrix, for a matrix ${\bf A}$, $\lambda_{\rm min}({\bf A})$ and $\lambda_{\rm max}({\bf A})$ denote respectively the smallest and largest eigenvalues of ${\bf A}$, while notation ${\rm vec}({\bf A})$ refers to the vector obtained by stacking the columns of ${\bf A}$. 
For ${\bf A}$, ${\bf B}$ two positive semi-definite matrices, ${\bf A}\preceq {\bf B}$ means that ${\bf B}-{\bf A}$ is positive semi-definite.
$X_n=o_p(1)$ implies the convergence in probability to zero of $X_n$ as $n$ goes to infinity and  $X_n=\mathcal{O}_p(1)$ implies that $X_n$ is bounded in probability. The arrow ``$\asto$" designates almost sure convergence while the arrow``$\xrightarrow[]{\mathcal{D}}$" refers to convergence in distribution. 
\section{Convergence of the regularized M-estimator of scatter matrix}
\label{sec:first_order}
Consider ${\bf x}_1,\cdots,{\bf x}_n$, $n$ observations of size $N$ defined as:
$$
{\bf x}_i=\boldsymbol{\Sigma}_N^{\frac{1}{2}}{\bf w}_i,
$$
where ${\bf w}_i\in\mathbb{C}^{N}$ are Gaussian zero-mean random vectors with covariance ${\bf I}_N$ and $\boldsymbol{\Sigma}_N\succeq 0$ is the population covariance matrix.
The regularized robust scatter estimator that will be considered in this work is that defined in \cite{Pascal-2013} as the unique solution $\hat{\bf C}_N(\rho)$ to:
\begin{equation}
\hat{\bf C}_N(\rho)=(1-\rho)\frac{1}{n}\sum_{i=1}^n \frac{{\bf x}_i{\bf x}_i^*}{\frac{1}{N}{\bf x}_i^*\hat{\bf C}_N^{-1}(\rho){\bf x}_i} + \rho {\bf I}_N.
\label{eq:hatc}
\end{equation}
with $\rho\in\left(\max(0,1-\frac{n}{N}),1\right]$\footnote{Another concurrent RTE is that of Chen \{\textit{et al} \cite{chen-11} which is given as the unique solution of
$$
\check{\bf C}_N(\rho)={\check{\bf B}_N(\rho)}{\frac{1}{N}\tr \check{\bf B}_N(\rho)}
$$ where
$$
\check{\bf B}_N(\rho)=(1-\rho)\frac{1}{n}\sum_{i=1}^n \frac{{\bf x}_i{\bf x}_i^*}{\frac{1}{N}{\bf x}_i^*\check{\bf C}_N(\rho)^{-1}{\bf x}_i} +\rho {\bf I}_N.
$$ 
}
Obviously, Chen's estimator is more involved and will not be thus considered in this work.  Such an estimator can be thought of as a hybrid robust-shrinkage estimator reminding Tyler's M-estimator of scale \cite{tyler} and Ledoit-Wolf's shrinkage estimator \cite{wolf}. It will be coined thus Regularized-Tyler estimator (RTE), and  defines a class of regularized-robust scatter estimators indexed by the regularization parameter $\rho$. When $n>N$, by varying $\rho$ from $0$ to $1$, one can move from the {unbiased Tyler-estimator} \cite{Pascal-08} to the identity matrix $(\rho=1)$ which corresponds to a trivial estimate of the unknown covariance matrix $\boldsymbol{\Sigma}{_N}$.

\subsection{Review of the results obtained in the large-$n,N$ regime}
Letting $c_N=\frac{N}{n}$, the large-$n,N$ regime will refer in the sequel to the one where $n\to\infty$ and $N\to\infty$ with $c_N\to c\in(0,\infty)$. 
 
As  mentioned earlier, unless considering particular assumptions on  $\boldsymbol{\Sigma}_N$, the {RTE} cannot be proven to converge (in any usual matrix norm) to  some deterministic matrix in the large-$n,N$ regime. Failing that, the approach pursued in \cite{couillet-kammoun-14} consists in  determining a random equivalent for the {RTE}, that corresponds to a standard matrix model. This finding is of utmost importance, since it allows one to replace the RTE, whose direct analysis is overly difficult, by another random object, for which an important load of results is available. The meaning of the equivalence between the {RTE} and the new object will be specified below.

Prior to presenting the results of \cite{couillet-kammoun-14}, we shall, for the reader convenience,  gather all the observations' properties in the following assumption:
\begin{assumption}
For $i\in\left\{1,\cdots,n\right\}$, ${\bf x}_i=\boldsymbol{\Sigma}_N^{\frac{1}{2}}{\bf w}_i$, with:
\begin{itemize}
\item ${\bf w}_1,\cdots,{\bf w}_n$ are $N\times1$ independent Gaussian random vectors with zero mean and covariance ${\bf I}_N$,
\item $\boldsymbol{\Sigma}_N\in\mathbb{C}^{N\times N}\succeq 0$ is such that $\frac{1}{N}\tr\boldsymbol{\Sigma}_N=1$.
\end{itemize}
\label{ass:model}
\end{assumption} 
It is worth noticing that the normalization $\frac{1}{N}\tr\boldsymbol{\Sigma}_N=1$ is considered for ease of exposition and is not limiting since the RTE is invariant to any scaling of $\boldsymbol{\Sigma}_N$. Denote by $\hat{\bf S}_N(\rho)$ the matrix given by:
$$
\hat{\bf S}_N(\rho)=\frac{1}{\gamma_N(\rho)}\frac{1-\rho}{1-(1-\rho)c_N}\frac{1}{n}\sum_{i=1}^n {\bf w}_i{\bf w}_i^* +\rho I_N,
$$
where $\gamma_N(\rho)$ is the unique positive solution to:
$$
1=\frac{1}{N}\tr \boldsymbol{\Sigma}_N\left(\rho\gamma_N(\rho)+(1-\rho)\boldsymbol{\Sigma}_N\right)^{-1}
$$
then $\hat{\bf S}_N(\rho)$ is equivalent to the RTE $\hat{\bf C}_N(\rho)$ in the sense of the following theorem,
\begin{theorem}
For any $\kappa >0$ small, define $\mathcal{R}_\kappa\triangleq\left[\kappa+\max(0,1-c^{-1}),1\right]$. Then, as $N,n\to\infty$ with $\frac{N}{n}\to c\in\left(0,\infty\right)$ and assuming $\lim\sup\|\boldsymbol{\Sigma}_N\|<\infty$,  we have:
$$
\sup_{\rho\in\mathcal{R}_\kappa}\left\|\hat{\bf C}_N(\rho)-\hat{\bf S}_N\right\|\asto 0.
$$
\label{th:large_nN}
\end{theorem}
\subsection{Convergence of the RTE in the large-$n$ regime}
In this section, we will consider the regime wherein $N$ is fixed and $n$ tends to infinity. An illustrative tool that is frequently used to handle this regime is  the strong law of large numbers {(SLLN)} which suggests replacing the average of independent and identically distributed random variables by their expected value.  This result should particularly serve to treat the term
$$
\frac{1}{n}\sum_{i=1}^n \frac{{\bf x}_i{\bf x}_i^*}{{\bf x}_i^*\hat{\bf C}_N^{-1}(\rho){\bf x}_i}
$$     
in the expression of the RTE. Nevertheless, because of the dependence of $\hat{\bf C}_N(\rho)$ on the observations ${\bf x}_i$, the {SLLN} cannot be directly applied to handle the above quantity. As we expect  $\hat{\bf C}_N(\rho)$ to converge to some deterministic matrix,  say $\boldsymbol{\Sigma}_0(\rho)$, it is sensible to substitute  $
\frac{1}{n}\sum_{i=1}^n \frac{{\bf x}_i{\bf x}_i^*}{{\bf x}_i^*\hat{\bf C}_N^{-1}(\rho){\bf x}_i}$ by $\frac{1}{n}\sum_{i=1}^n \frac{{\bf x}_i{\bf x}_i^*}{{\bf x}_i^*\boldsymbol{\Sigma}_0^{-1}(\rho){\bf x}_i}$. The latter quantity   is in turn equivalent to $\mathbb{E}\left[\frac{{\bf x}{\bf x}^*}{{\bf x}^*\boldsymbol{\Sigma}_0^{-1}(\rho){\bf x}}\right]$ from the {SLLN} where the expectation is taken over the distribution of the random vectors ${\bf x}_i$. 
Based on these heuristic arguments, a plausible  guess  is that $\hat{\bf C}_N(\rho)$  converges to $\boldsymbol{\Sigma}_0(\rho)$, the solution to the following equation:
\begin{equation}
\boldsymbol{\Sigma}_0(\rho)=N(1-\rho)\mathbb{E}\left[\frac{{\bf x}{\bf x}^*}{{\bf x}^*\boldsymbol{\Sigma}_0^{-1}(\rho){\bf x}}\right]+\rho {\bf I}_N.
\label{eq:sigma_0}
\end{equation}
The main goal of this section is to establish the convergence of $\hat{\bf C}_N(\rho)$ to $\boldsymbol{\Sigma}_0(\rho)$. We will assume that $\boldsymbol{\Sigma}_0(\rho)$ exists for each $\rho\in\left(0,1\right]$. The existence and uniqueness of  $\boldsymbol{\Sigma}_0(\rho)$ will be discussed later on in this section.  
Similar to the large-$n,N$  regime, we need to introduce a random equivalent for $\hat{\bf C}_N(\rho)$ that is easier to handle. Naturally,  an intuitive random equivalent is obtained by replacing, in the right-hand side of \eqref{eq:hatc}, $\hat{\bf C}_N(\rho)$ by $\boldsymbol{\Sigma}_0(\rho)$, thus yielding:
\begin{equation}
\tilde{\boldsymbol{\Sigma}}(\rho)=N(1-\rho)\frac{1}{n}\sum_{i=1}^n \frac{{\bf x}_i{\bf x}_i^*}{{\bf x}_i^*\boldsymbol{\Sigma}_0^{-1}(\rho){\bf x}_i}+\rho {\bf I}_N.
\label{eq:tilde_sigma}
\end{equation}
Unlike $\hat{\bf C}_N(\rho)$, $\tilde{\boldsymbol{\Sigma}}(\rho)$ is more tractable, being an explicit function of the  observations' vectors. By the {SLLN}, $\tilde{\boldsymbol{\Sigma}}(\rho)$ is an unbiased estimate of $\boldsymbol{\Sigma}_0(\rho)$ that satisfies:
$$
\boldsymbol{\Sigma}_0(\rho)=\tilde{\boldsymbol{\Sigma}}(\rho)+\boldsymbol{\epsilon}_n(\rho),
$$
where $\boldsymbol{\epsilon}_n(\rho)$ is an $N\times N$ matrix whose elements converge almost surely to zero and are bounded in probability at the rate $\frac{1}{n}$, i.e, 
$$\left[\boldsymbol{\epsilon}_n(\rho)\right]_{i,j}=\mathcal{O}_p\left(\frac{1}{n}\right).$$

For the above convergence to hold uniformly in $\rho$, one needs to check that the first absolute second moment of the entries of $\frac{{\bf x}{\bf x}^*}{{\bf x}^*\boldsymbol{\Sigma}_0^{-1}(\rho){\bf x}}$ is uniformly bounded in $\rho$. To this end we shall additionally assume that:
 \begin{assumption}
\item Matrix $\boldsymbol{\Sigma}_N$ is non-singular, i.e., the smallest eigenvalue of $\boldsymbol{\Sigma}_N$, $\lambda_{\rm min}(\boldsymbol{\Sigma}_N)$ satisfies:
$$
\lambda_{\rm min}(\boldsymbol{\Sigma}_N) >0.
$$
\label{ass:min}
\end{assumption}
Under Assumption \ref{ass:min},  the spectral norm of $\boldsymbol{\Sigma}_0(\rho)$ can be bounded as:
\begin{lemma}
Let $\boldsymbol{\Sigma}_0$ be the solution to \eqref{eq:sigma_0}, whenever it exists. Then,
$$
\sup_{\rho\in\left[\kappa,1\right]}\left\|\boldsymbol{\Sigma}_0(\rho)\right\| \leq \frac{\|\boldsymbol{\Sigma}_N\|}{\lambda_{\rm min}(\boldsymbol{\Sigma}_N)}
$$
where $\kappa>0$ is some positive scalar.  
\label{lemma:bounded_spectral}
\end{lemma}
\begin{proof}
See Appendix \ref{app:bounded_spectral}
\end{proof}
Equipped with the bound provided by Lemma \ref{lemma:bounded_spectral}, we can claim that:
$$
\sup_{\rho\in \left[\kappa,1\right]}\left|\left[\boldsymbol{\epsilon}_n(\rho)\right]_{i,j}\right|=\mathcal{O}_p\left(\frac{1}{n}\right)
$$ 
or equivalently:
$$
\sup_{\rho\in \left[\kappa,1\right]} \left\|\tilde{\boldsymbol{\Sigma}}(\rho)-\boldsymbol{\Sigma}_0(\rho)\right\| =\mathcal{O}_p\left(\frac{1}{n}\right).
$$
Characterizing the rate of convergence of $\tilde{\boldsymbol{\Sigma}}(\rho)$ to $\boldsymbol{\Sigma}_0(\rho)$ is of fundamental importance and would later help in  the derivation of the second-order statistics for $\tilde{\boldsymbol{\Sigma}}(\rho)$ and then for $\hat{\bf C}_N(\rho)$. 

Before stating our first main result, we would like to particularly stress the fact that  Assumption \ref{ass:min} is not limiting. To see that, consider $\boldsymbol{\Sigma}_N={\bf U}\boldsymbol{\Lambda}{\bf U}^*$ the eigenvalue decomposition of $\boldsymbol{\Sigma}_N$ wherein the diagonal elements of $\boldsymbol{\Lambda}$, $\lambda_1,\cdots,\lambda_N$  correspond to the eigenvalues of $\boldsymbol{\Sigma}_N$ arranged in the decreasing order, i.e., $\lambda_1 \geq \lambda_2 \cdots \geq \lambda_N$. Denoting by $r$ the rank of $\boldsymbol{\Sigma}_N$, then, $\lambda_{r+1}=\cdots=\lambda_N=0$. Write ${\bf U}$ as ${\bf U}=\left[{\bf U}_{r},{\bf U}_{N-r}\right]$, ${\bf U}_r\in\mathbb{C}^{N\times r}$. Then, it is easy to see that:
$$
{\hat{\bf C}_N{(\rho)}}{\bf U}_{N-r}=\rho{\bf U}_{N-r}
$$ 
while:
\begin{equation}
{\bf U}_r^*{\hat{\bf C}_N}(\rho){\bf U}_r=(1-\rho)\frac{1}{n}\sum_{i=1}^n \frac{\boldsymbol{\Lambda_r}^{\frac{1}{2}}\tilde{\bf w}_i\tilde{\bf w}_i^*\boldsymbol{\Lambda}_r^{\frac{1}{2}}}{\frac{1}{N}\tilde{\bf w}_i^*\boldsymbol{\Lambda}_r^{\frac{1}{2}}{\bf U}_r^*\hat{\bf C}_N^{-1}(\rho){\bf U}_r\boldsymbol{\Lambda}_r^{\frac{1}{2}}\tilde{\bf w}_i} +\rho {\bf I}_N,
\label{eq:C_N}
\end{equation}
where $\tilde{\bf w}_i={\bf U}_r^*{\bf x}_i$ follows a Gaussian distribution with zero-mean and covariance ${\bf I}_r$. 
Since $\left({\bf U}_r^*\hat{\bf C}_N(\rho){\bf U}_r\right)^{-1}={\bf U}_r^*\hat{\bf C}_N^{-1}(\rho){\bf U}_r$, instead of using $\hat{\bf C}_N(\rho)$, it thus suffices to work with  ${\bf U}_r^*{\hat{\bf C}_N}(\rho){\bf U}_r$, for which Assumption \ref{ass:min} can be used.

The following theorem establishes the convergence of   $\hat{\bf C}_N(\rho)$  to $\boldsymbol{\Sigma}_0(\rho)$, the hypothetical solution to \eqref{eq:sigma_0}, 

\begin{theorem}
Assume that there exists a unique solution $\boldsymbol{\Sigma}_0(\rho)$ to \eqref{eq:sigma_0}. 
Let $\kappa>0$ be {a} some small positive real {scalar}. Then, assuming that {Assumptions \ref{ass:model} and \ref{ass:min}} hold true, {one has} under the large-$n$ regime:
$$
\sup_{\rho\in\left[\kappa,1\right]}\left\|\hat{\bf C}_N(\rho)- \boldsymbol{\Sigma}_0(\rho)\right\|\asto 0.
$$
Moreover, 
$$
\sup_{\rho\in\left[\kappa,1\right]} \left\|\hat{\bf C}_N(\rho)-\boldsymbol{\Sigma}_0(\rho)\right\|=\mathcal{O}_p\left(\frac{1}{n}\right).
$$
\label{th:first_order}
\end{theorem} 
\begin{proof}
See Appendix \ref{app:first_order}
\end{proof}
In Theorem \ref{th:first_order}, we establish the convergence of $\hat{\bf C}_N(\rho)$ to some limiting matrix $\boldsymbol{\Sigma}_0(\rho)$ that solves the fixed point equation \eqref{eq:sigma_0}. While \eqref{eq:sigma_0} seems to fully characterize $\boldsymbol{\Sigma}_0(\rho)$, it does not clearly unveil its relationship with the observations' covariance matrix $\boldsymbol{\Sigma}_N$.  The major intricacy stems from the expectation operator in the term $\mathbb{E}\left[\frac{{\bf x}{\bf x}^*}{{\bf x}^*\boldsymbol{\Sigma}_0^{-1}(\rho){\bf x}}\right]$. A close look to this quantity reveals that it can be further developed by leveraging some interesting features of Gaussian distributed vectors. Note first that \eqref{eq:sigma_0} is also equivalent to:
\begin{equation}
N(1-\rho)\mathbb{E}\left[\frac{{\bf w}{\bf w}^*}{{\bf w}^*\boldsymbol{\Sigma}_N^{\frac{1}{2}}\boldsymbol{\Sigma}_0^{-1}(\rho)\boldsymbol{\Sigma}_N^{\frac{1}{2}}{\bf w}}\right]+\rho\boldsymbol{\Sigma}_N^{-1}=\boldsymbol{\Sigma}_N^{-\frac{1}{2}}\boldsymbol{\Sigma}_0(\rho)\boldsymbol{\Sigma}_N^{-\frac{1}{2}},
\label{eq:sigma_0_transform}
\end{equation}
where ${\bf w}\sim\mathcal{CN}({\bf 0},{\bf I}_N)$. Let $\boldsymbol{\Sigma}_N^{\frac{1}{2}}\boldsymbol{\Sigma}_0^{-1}(\rho)\boldsymbol{\Sigma}_N^{\frac{1}{2}}={\bf V}{\bf D}{\bf V}^*$  be an eigenvalue decomposition of  $\boldsymbol{\Sigma}_N^{\frac{1}{2}}\boldsymbol{\Sigma}_0^{-1}(\rho)\boldsymbol{\Sigma}_N^{\frac{1}{2}}$, where ${\bf D}$ is a diagonal matrix with diagonal elements $d_1,d_2, \cdots,d_N$. {Notice that, of course the $d_i$'s depend on $\rho$. However, for simplicity purposes, the notation with $(\rho)$ is omitted.} Since the Gaussian distribution is invariant under unitary transformation,  \eqref{eq:sigma_0_transform} is also equivalent to:
\begin{equation}
N(1-\rho)\mathbb{E}\left[\frac{{\bf w}{\bf w}^*}{{\bf w}^*{\bf D}{\bf w}}\right] +\rho {\bf V}^*\boldsymbol{\Sigma}_N^{-1}{\bf V}={\bf D}^{-1}.
\label{eq:equivalent}
\end{equation}
It is not difficult to see that the off-diagonal elements of  $\mathbb{E}\left[\frac{{\bf w}{\bf w}^*}{{\bf w}^*{\bf D}{\bf w}}\right]$ are equal to zero. In effect for $i\neq j$, writing $w_i$ as $r_ie^{\jmath \theta_i}$ with $r_i$ {Rayleigh} distributed and $\theta_i$ independent of $r_i$ and uniformly distributed over $[-\pi,\pi]$, {one has} $\mathbb{E}\left[\left[\frac{{\bf w}{\bf w}^*}{{\bf w}^*{\bf D}{\bf w}}\right]_{i,j}\right]=\mathbb{E}\left[\frac{r_ir_j^*e^{\jmath(\theta_i-\theta_j)}}{\sum_{i=1}^N d_i |r_i|^2}\right]$ which can be shown to be zero by taking the expectation over the difference of phase $\theta_i-\theta_j$.
Therefore,  $\mathbb{E}\left[\frac{{\bf w}{\bf w}^*}{{\bf w}^*{\bf D}{\bf w}}\right]$ is diagonal, with diagonal elements $\left(\alpha_i\right)_{i=1,\cdots,N}$ given by:
$$
\alpha_i({\bf D})=\mathbb{E}\left[\frac{|w_i|^2}{{\bf w}^*{\bf D}{\bf w}}\right].
$$
Hence, ${\bf V}^*\boldsymbol{\Sigma}_N^{-1}{\bf V}$ is also diagonal, thus implying that $\boldsymbol{\Sigma}_N$ and ${\boldsymbol{\Sigma}_0(\rho)}$ share the same eigenvector matrix ${\bf U}$. 
In order to prove the existence of $\boldsymbol{\Sigma}_0(\rho)$, it suffices to check that {$d_1,\cdots,d_N$ are} solutions to the following equation:
\begin{equation}
N(1-\rho)\alpha_i({\bf D}) +\frac{\rho}{\lambda_i} =\frac{1}{d_i}.
\label{eq:system}
\end{equation}
To this end, consider 
\begin{align*}
&h:\mathbb{R}_{+}^{N}\to \mathbb{R}_{+}^{N}\\
& \left(x_1,\cdots,x_N\right) \mapsto N(1-\rho)\left(\mathbb{E}\left[\frac{|w_1|^2}{\sum_{j=1}^N \frac{1}{x_j}|w_j|^2}\right]+\frac{\rho}{\lambda_1},\cdots,\right.\\
&\left.\mathbb{E}\left[\frac{|w_N|^2}{\sum_{j=1}^N \frac{1}{x_j}|w_j|^2}\right]+\frac{\rho}{\lambda_N}\right).
 \end{align*}
Proving that $d_1,\cdots,d_N$ are the unique solutions of \eqref{eq:system} is equivalent to showing that:
\begin{equation}
{\bf x}=h\left(x_1,\cdots,x_N\right)
\label{eq:fixed}
\end{equation}
admits a unique positive solution. For this, we show that $h$ satisfies the following properties:
\begin{itemize}
\item Nonnegativity: For each $x_1,\cdots,x_N\geq 0$, vector $h(x_1,\cdots,x_N)$has positive elements.
\item Monotonicity: For each $x_1\geq x_1^{'},\cdots,x_N\geq x_N^{'}$, $h(x_1,\cdots,x_N)\geq h(x_1^{'},\cdots,x_N^{'})$ where $\geq$ holds element-wise. 
\item Scalability: For each $\alpha>1$, $\alpha h(x_1,\cdots,x_N) > h(\alpha x_1,\cdots,\alpha x_N)$.
\end{itemize}
The first item is trivial. The second one follows from the fact that $h$ is an increasing function of each $x_i$. As for the last item, it follows by noticing that as $\rho>0$, 
$$
\mathbb{E}\left[\frac{|w_i|^2}{\sum_{j=1}^N \frac{1}{\alpha x_j}|w_j|^2}\right]+\frac{\rho}{\lambda_j} <\alpha\left(\mathbb{E}\left[\frac{|w_i|^2}{\sum_{j=1}^N \frac{1}{ x_j}|w_j|^2}\right]+\frac{\rho}{\lambda_j}\right)
$$
According to \cite{YAT95}, $h$ is  a standard interference function, and if there exists $q_1,\cdots,q_N$ such that ${\bf q} > h(q_1,\cdots,q_N)$ where $>$ holds element-wise,  then there is a unique ${\bf x}_{\infty}=\left(x_{1,\infty},\cdots,x_{N,\infty}\right)$ such that:
$$
{\bf x}_{\infty}=h(x_{1,\infty},\cdots,x_{N,\infty}).
$$
Moreover, ${\bf x}_{\infty}=\lim_{t\to\infty}{\bf x}^{(t)}$ with ${\bf x}^{(0)}> 0$ arbitrary and for $t\geq 0$, ${\bf x}^{(t+1)}=h(x_1^{(t)},\cdots,x_N^{(t)}) $. 
To prove the feasibility condition, take ${\bf q}=\left(q,\cdots,q\right)$. Then, ${h(q,\cdots,q)}=(1-\rho)q+\frac{\rho}{\lambda_i}$. Setting $q\geq \frac{1}{\lambda_{\rm min}}$, we get that ${h(q,\cdots,q)}<{\bf q}$, thereby establishing the desired inequality.

The interest of the framework of Yates \cite{YAT95} is that in addition to being a useful tool for proving existence and uniqueness of the fixed-point of a standard interference function, it shows that the solution can be numerically approximated by computing iteratively ${\bf x}^{(t+1)}=h(x_1^{t},\cdots,x_N^{t})$. 
However, in order to implement this algorithm, one needs to further develop the terms $\alpha_i({\bf D})$. This is in particular the goal of the following lemma, the proof of which is deferred to Appendix~\ref{app:di}.
\begin{lemma}
Let ${\bf w}=\left[w_1,\cdots,w_N\right]^{\mbox{\tiny T}}$ be a standard complex Gaussian  vector and ${\bf D}=\left(d_1,\cdots,d_N\right)$ be a diagonal matrix with positive diagonal elements. Consider $\alpha_1,\cdots,\alpha_N$, the set of scalars given by:
$$
\alpha_i({\bf D})=\mathbb{E}\left[\frac{|w_i|^2}{\sum_{i=1}^N d_i |w_i|^2}\right].
$$
Then
\begin{align*}
&\alpha_i({\bf D})=\frac{1}{2^NN}\frac{1}{d_i\prod_{j=1}^N d_j}\\
&\times F_D^{(N)}\left(N,1,\cdots,\!\underset{\substack{\uparrow \\ i\textnormal{-th} \\ \textnormal{position}}}{2},\!1,\cdots\!,1,N+1,\frac{d_1-\frac{1}{2}}{d_1},\cdots,\!\frac{d_N-\frac{1}{2}}{d_N}\right),
\end{align*}
where $F_D^{(N)}$ is the Lauricella's type $D$ hypergeometric function.~\footnote{The evaluation of the Lauricella's type $D$ hypergeometric function is performed numerically using its integral representation
\begin{align*}
&F_D^{(N)}(a,b_1,\cdots,b_n,c;x_1,\cdots,x_n)\\
&=\frac{\Gamma(c)}{\Gamma(a)\Gamma(a-c)}\int_0^1 t^{a-1}(1-t)^{c-a-1}\prod_{i=1}^n (1-x_it)^{-b_i}dt. \hspace{0.1cm} \Re c > \Re a >0.
\end{align*}}
\label{lemma:di}
\end{lemma}
Equipped with the result of Lemma \ref{lemma:di}, we will now show how one can in practice approximate  $\boldsymbol{\Sigma}_0(\rho)$. First, one needs to approximate the solution of \eqref{eq:fixed}. Let ${\bf d}^{0}=\left[d_1^{(0)},\cdots,d_N^{(0)}\right]^{\mbox{\tiny T}}$ be an arbitrary vector with positive elements. We set ${\bf d}^{(t)}=\left[d_1^{(t)},\cdots,d_N^{(t)}\right]$ as:
$$
d_i^{(t+1)}= \frac{1}{\frac{\rho }{\lambda_i}+N(1-\rho)\alpha_i({\rm diag}({\bf d}^{(t)}))}
$$
where the expression of $\alpha_i({\rm diag}({\bf d}^{(t)}))$ is given by Lemma \ref{lemma:di}.
As $t\to\infty$,  ${\bf d}^{(t)}$ tends to ${\bf d}$, the vector of eigenvalues of  $\boldsymbol{\Sigma}_N^{\frac{1}{2}}\boldsymbol{\Sigma}_0^{-1}(\rho)\boldsymbol{\Sigma}_N^{\frac{1}{2}}$ which is the solution of \eqref{eq:fixed}. Since $\boldsymbol{\Sigma}_N$ and $\boldsymbol{\Sigma}_0(\rho)$ share the same eigenvectors, the eigenvalues $s_{1,\infty},\cdots,s_{N,\infty}$ of  $\boldsymbol{\Sigma}_0(\rho)$ are given by
$
s_{i,\infty}=\frac{\lambda_i}{d_{i}}.
$
The matrix $\boldsymbol{\Sigma}_0(\rho)$ is finally given by:
$$
\boldsymbol{\Sigma}_0(\rho)={\bf U}\hspace{0.05cm}{\rm diag}(\left[s_{1,\infty},\cdots,s_{N,\infty}\right]){\bf U}^*.
$$
While the above characterization of $\boldsymbol{\Sigma}_0{(\rho)}$ seems to provide few insights in most cases, it shows that except for  the particular case of $\boldsymbol{\Sigma}_N={\bf I}_N$, the RTE $\hat{\bf C}_N(\rho)$ is biased  for $\rho\in\left[\kappa,1\right)$ in that:
$$
\boldsymbol{\Sigma}_0{(\rho)} \neq \boldsymbol{\Sigma}_N.
$$
To see that, notice that $\boldsymbol{\Sigma}_0{(\rho)} = \boldsymbol{\Sigma}_N$ implies that ${\bf D}={\bf I}_N$. Replacing ${\bf D}$ by the identity matrix in \eqref{eq:sigma_0_transform} and using the fact that $\mathbb{E}\left[\frac{{\bf w}{\bf w}^*}{{\bf w}^*{\bf w}}\right]=\frac{1}{N}{\bf I}_N$ shows that only $\boldsymbol{\Sigma}_N={\bf I}_N$ satisfies a null bias. 
{Hence}, it appears that improving the conditioning of the RTE  by using a non-zero regularization coefficient comes in general at the cost of a higher bias. 
\section{Second order statistics in the large-$n$ regime}
\label{sec:second_order}
The previous section establishes the convergence of the RTE to the limiting deterministic matrix $\boldsymbol{\Sigma}_0(\rho)$. {In the following, for readability purposes, $\boldsymbol{\Sigma}_0(\rho)$ will be simply replaced by $\boldsymbol{\Sigma}_0$.}  The convergence holds in the almost sure sense, and can help infer the asymptotic limit of any functional of the RTE. More formally, for any  functional $f$ continuous around $\boldsymbol{\Sigma}_0$, $f(\hat{\bf C}_N)$ converges almost surely to $f(\boldsymbol{\Sigma}_0)$. While this result can be used to understand the convergence of inference methods using RTEs,  it becomes of little help when one is required to deeply  understand their fluctuations, a prerequisite that essentially arises in many detection applications. This motivates the present section which aims at establishing a Central Limit Theorem (CLT) for the RTE.

It is worth noticing that the scope of applicability of the results obtained in the large-$n$ regime is much wider than that of the $n,N$ large regime. As a matter of fact, using the Delta Method \cite{vaart}, our result can help obtain the CLT for any continuous functional of the RTE. We deeply believe that this can facilitate the design of inference methods using RTEs. 

Although treatments of both regimes seem to take different directions, they have thus far presented the common denominator of relying on an intermediate random equivalent for $\hat{\bf C}_N(\rho)$, be it $\tilde{\boldsymbol{\Sigma}}(\rho)$ or $\hat{\bf S}_N(\rho)$ (See Theorem \ref{th:large_nN}). It is thus easy to convince oneself that in order to derive the CLT for $\hat{\bf C}_N(\rho)$, a CLT for $\tilde{\boldsymbol{\Sigma}}(\rho)$ is required. 

We denote in the sequel by $\boldsymbol{\delta}$ and $\tilde{\boldsymbol{\delta}}$ the quantities:
$\boldsymbol{\delta}={\rm vec}(\hat{\bf C}_N(\rho))-{\rm vec}(\boldsymbol{\Sigma}_0)$ and  $\tilde{\boldsymbol{\delta}}={\rm vec}(\tilde{\boldsymbol{\Sigma}}(\rho))-{\rm vec}(\boldsymbol{\Sigma}_0)$ and consider the derivation of the CLT for vectors $\boldsymbol{\delta}$ and then for $\tilde{\boldsymbol{\delta}}$. We will particularly prove that $\boldsymbol{\delta}$ and $\tilde{\boldsymbol{\delta}}$ behave in the large-$n$ regime as Gaussian random variables that can be fully characterized by their covariance matrices $\mathbb{E}\left[\boldsymbol{\delta}\boldsymbol{\delta}^*\right]$ and $\mathbb{E}[\boldsymbol{\tilde{\delta}}\boldsymbol{\tilde{\delta}}^*]$. Starting with the observation that in many signal processing applications, the focus might be put on the second-order statistics of the  real and imaginary parts of $\boldsymbol{\delta}$ and $\tilde{\boldsymbol{\delta}}$, we additionally provide expressions for the pseudo-covariance matrices  $\mathbb{E}\left[\boldsymbol{\delta}\boldsymbol{\delta}^{\mbox{\tiny T}}\right]$ and $\mathbb{E}[\boldsymbol{\tilde{\delta}}\boldsymbol{\tilde{\delta}}^{\mbox{\tiny T}}]$ of $\boldsymbol{\delta}$ and $\tilde{\boldsymbol{\delta}}$ which, coupled with that of covariance matrices, suffice to fully characterize fluctuations of the vectors $\left[\Re \boldsymbol{\delta}^{\mbox{\tiny T}},\Im \boldsymbol{\delta}^{\mbox{\tiny T}}\right]^{\mbox{\tiny T}}$ and $[\Re \tilde{\boldsymbol{\delta}}^{\mbox{\tiny T}},\Im \tilde{\boldsymbol{\delta}}^{\mbox{\tiny T}}]^{\mbox{\tiny T}}$.

We will start by handling the fluctuations of $\tilde{\boldsymbol{\delta}}$. 
To this end, we need first to work out the expression of $\tilde{\boldsymbol{\Sigma}}(\rho)$. Recall that $\tilde{\boldsymbol{\Sigma}}(\rho)$ is given by:
$$
\tilde{\boldsymbol{\Sigma}}(\rho)=\frac{N(1-\rho)}{n}\sum_{i=1}^n \frac{{\bf x}_i{\bf x}_i^*}{{\bf x}_i^*{\boldsymbol{\Sigma}_0^{-1}}{\bf x}_i}+\rho {\bf I}_N.
$$
Therefore,
\begin{align*}
\boldsymbol{\Sigma}_0^{-\frac{1}{2}}\tilde{\boldsymbol{\Sigma}}(\rho)\boldsymbol{\Sigma}_0^{-\frac{1}{2}}-{\bf I}_N&=\frac{N(1-\rho)}{n}\sum_{i=1}^n \frac{\boldsymbol{\Sigma}_0^{-\frac{1}{2}}\boldsymbol{\Sigma}_N^{\frac{1}{2}}{\bf w}_i{\bf w}_i^*\boldsymbol{\Sigma}_N^{\frac{1}{2}}\boldsymbol{\Sigma}_0^{-\frac{1}{2}}}{{\bf w}_i^*\boldsymbol{\Sigma}_N^{\frac{1}{2}}{\boldsymbol{\Sigma}_0^{-1}}\boldsymbol{\Sigma}_N^{\frac{1}{2}}{\bf w}_i}\\
&+\rho \boldsymbol{\Sigma}_0^{-1}-{\bf I}_N
\end{align*}
Using  the eigenvalue decomposition of $\boldsymbol{\Sigma}_N^{\frac{1}{2}}\boldsymbol{\Sigma}_0^{-1}\boldsymbol{\Sigma}_N^{\frac{1}{2}}={\bf U}{\bf D}{\bf U}^*$ and denoting $\tilde{\bf w}_i={\bf U}^*{\bf w}_i$, we thus obtain:
\begin{align*}
{\bf U}^*\boldsymbol{\Sigma}_0^{-\frac{1}{2}}\tilde{\boldsymbol{\Sigma}}(\rho)\boldsymbol{\Sigma}_0^{-\frac{1}{2}}{\bf U}-{\bf I}_N&=\frac{N(1-\rho)}{n}\sum_{i=1}^n \frac{{\bf D}^{\frac{1}{2}}\tilde{\bf w}_i\tilde{\bf w}_i^*{\bf D}^{\frac{1}{2}}}{\tilde{\bf w}_i^*{\bf D}\tilde{\bf w}_i}\\
&+\rho {\bf U}^*\boldsymbol{\Sigma}_0^{-1}{\bf U}-{\bf I}_N.
\end{align*}
From the characterization of ${\boldsymbol{\Sigma}_0}$ provided in the previous section, we can easily check that:
$$
N(1-\rho) \mathbb{E}\left[\frac{{\bf D}^{\frac{1}{2}}\tilde{\bf w}\tilde{\bf w}^*{\bf D}^{\frac{1}{2}}}{\tilde{\bf w}^*{\bf D}\tilde{\bf w}}\right] ={\bf I}_N-\rho {\bf U}^*\boldsymbol{\Sigma}_0^{-1}{\bf U}
$$
Therefore,
\begin{align}
&{\bf U}^*\boldsymbol{\Sigma}_0^{-\frac{1}{2}}\tilde{\boldsymbol{\Sigma}}(\rho)\boldsymbol{\Sigma}_0^{-\frac{1}{2}}{\bf U}-{\bf I}_N\\
&=\frac{N(1-\rho)}{n}\sum_{i=1}^n \left[\frac{{\bf D}^{\frac{1}{2}}\tilde{\bf w}_i\tilde{\bf w}_i^*{\bf D}^{\frac{1}{2}}}{\tilde{\bf w}_i^*{\bf D}\tilde{\bf w}_i}\
-\mathbb{E}\left[\frac{{\bf D}^{\frac{1}{2}}\tilde{\bf w}\tilde{\bf w}^*{\bf D}^{\frac{1}{2}}}{\tilde{\bf w}^*{\bf D}\tilde{\bf w}}\right]\right].\label{eq:latter}
\end{align}
From \eqref{eq:latter}, it appears that the asymptotic distribution of $[\Re\boldsymbol{\tilde{\delta}}^{\mbox{\tiny T}},\Im\boldsymbol{\tilde{\delta}}^{\mbox{\tiny T}}]^{\mbox{\tiny T}}$ is Gaussian and thus can be fully characterized by its asymptotic covariance and pseudo-covariance matrices. 
  Using \eqref{eq:latter}, it is easy to show that we need for that the pseudo-covariance  and covariance matrices of:
$$
\frac{1}{n}\sum_{i=1}^n \frac{{\rm vec}(\tilde{\bf w}_i\tilde{\bf w}_i^*)}{\tilde{\bf w}_i^*{\bf D}\tilde{\bf w}_i}-\mathbb{E}\left[\frac{{\rm vec}(\tilde{\bf w}\tilde{\bf w}^*)}{\tilde{\bf w}^*{\bf D}\tilde{\bf w}}\right].
$$
These quantities involve the following set of scalars, 
$$
\beta_{i,j}=\mathbb{E}\left[\frac{|w_i|^2|w_j|^2}{\left({\bf w}^*{\bf D}{\bf w}\right)^2}\right] \hspace{0.2cm} i,j=1,\cdots,N
$$
for which  closed-form expressions need to be derived. This is the objective of the following technical lemma, which is of independent interest:
\begin{lemma}
\label{lemma:beta}
Let ${\bf w}=\left[w_1,\cdots,w_N\right]^{\mbox{\tiny T}}$ be a standard complex Gaussian vector and ${\bf D}={\rm diag}(d_1,\cdots,d_N)$ be a diagonal matrix with positive diagonal elements. Consider $\beta_{i,j}$ as above. 
Then $\beta_{i,j}$ are given for $i=j$ and $i\neq j$ by the expressions in \eqref{eq:betaii}, \eqref{eq:betaij} and \eqref{eq:betaji} at the top of {next} page.
\begin{figure*}
\begin{align}
\beta_{i,i}&=\frac{1}{2^{N-1}N(N+1)}\frac{1}{d_i^2\prod_{k=1}^N d_k} F_{D}^{N}\left(N,1\cdots,1,\underset{\substack{\uparrow \\ i\textnormal{-th} \\ \textnormal{position}}}{3},1,\cdots,1,N+2,\frac{d_1-\frac{1}{2}}{d_1},\cdots,\frac{d_N-\frac{1}{2}}{d_N}\right) \label{eq:betaii}\\
\beta_{i,j}&=\frac{1}{2^{N}N(N+1)}\frac{1}{d_i d_j\prod_{k=1}^N d_k}F_{D}^{N}\left(N,1\cdots,1,\underset{\substack{\uparrow \\ i\textnormal{-th} \\ \textnormal{position}}}{2},1,\cdots,1,\underset{\substack{\uparrow \\ j\textnormal{-th} \\ \textnormal{position}}}{2},1\cdots,1,N+2,\frac{d_1-\frac{1}{2}}{d_1},\cdots,\frac{d_N-\frac{1}{2}}{d_N}\right), i<j\label{eq:betaij}\\
\beta_{i,j}&=\beta_{j,i}, \hspace{0.1cm} i>j\label{eq:betaji}
\end{align}
\hrule
\end{figure*}
\end{lemma}
With this result at hand, the next Lemma follows immediately:
\begin{lemma}
Let ${\bf D}$ be $N\times N$ diagonal matrix with positive diagonal elements. Consider $\tilde{\bf w}_1,\cdots,\tilde{\bf w}_n$ $n$ {independent} complex Gaussian random vectors with zero-mean and covariance ${\bf I}_N$. Then, $\sqrt{n}\left(\frac{1}{n}\sum_{i=1}^n \frac{{\rm vec}(\tilde{\bf w}_i\tilde{\bf w}_i^*)}{\tilde{\bf w}_i^*{\bf D}\tilde{\bf w}_i}-\mathbb{E}\left[\frac{{\rm vec}(\tilde{\bf w}\tilde{\bf w}^*)}{\tilde{\bf w}^*{\bf D}\tilde{\bf w}}\right]\right)$ converges to a multivariate Gaussian distribution with covariance ${\bf B}({\bf D})$ and pseudo-covariance ${\bf G}({\bf D})$ given by:
\begin{align}
{\bf B}({\bf D})&=\tilde{\bf B}({\bf D})-{\rm vec}(\boldsymbol{\Xi}){\rm vec}(\boldsymbol{\Xi})^{\mbox{\tiny T}} \label{eq:B}\\
{\bf G}({\bf D})&=\tilde{\bf G}({\bf D})-{\rm vec}(\boldsymbol{\Xi}){\rm vec}(\boldsymbol{\Xi})^{\mbox{\tiny T}} \label{eq:G}
\end{align}
where 
\begin{align*}
\tilde{\bf B}({\bf D})&=\mathbb{E}\left[\frac{{\rm vec}(\tilde{\bf w}\tilde{\bf w}^*)\left({\rm vec}(\tilde{\bf w}\tilde{\bf w}^*))\right)^{*}}{(\tilde{\bf w}^*{\bf D}\tilde{\bf w})^2}\right]\\
\tilde{\bf G}({\bf D})&=\mathbb{E}\left[\frac{{\rm vec}(\tilde{\bf w}\tilde{\bf w}^*)\left({\rm vec}(\tilde{\bf w}\tilde{\bf w}^*))\right)^{\mbox{\tiny T}}}{(\tilde{\bf w}^*{\bf D}\tilde{\bf w})^2}\right]\\
\boldsymbol{\Xi}({\bf D})&={\rm diag}\left(\alpha_1({\bf D}),\cdots,\alpha_N({\bf D})\right)
\end{align*}
Furthermore, $\tilde{\bf B}$ and $\tilde{\bf G}$ are composed of $N^2$ block of $N\times N$ matrices, i.e, $\tilde{\bf B}({\bf D})=\begin{bmatrix}
\tilde{\bf B}_{1,1} &\cdots& \tilde{\bf B}_{1,N}\\
&\ddots & \\
\tilde{\bf B}_{N,1}&\cdots & \tilde{\bf B}_{N,N}
\end{bmatrix}
$,
 $\tilde{\bf G}({\bf D})=\begin{bmatrix}
\tilde{\bf G}_{1,1} &\cdots& \tilde{\bf G}_{1,N}\\
&\ddots & \\
\tilde{\bf G}_{N,1}&\cdots & \tilde{\bf G}_{N,N}
\end{bmatrix}
$
where: 
\begin{align*}
\tilde{{\bf B}}_{i,i}&={\rm diag}\left(\beta_{i,1}\cdots,\beta_{i,N}\right)\\
\left[\tilde{{\bf B}}_{i,j}\right]_{k,\ell}&=1_{\left\{k=i,\ell=j\right\}}\beta_{i,j}, \hspace{0.1cm} i\neq j\\
\left[\tilde{{\bf G}}_{i,j}\right]_{k,\ell}&=1_{\left\{k=i,\ell=j\right\}}\beta_{i,j}+1_{\left\{k=j,\ell=i\right\}}\beta_{i,j}.
\end{align*}
\label{lemma:covariance}
\end{lemma}
Equipped with Lemma \ref{lemma:covariance}, we are now in position to state the CLT for $\tilde{\boldsymbol{\Sigma}}(\rho)$, whose proof is omitted being a direct consequence of Lemma~\ref{lemma:covariance}:
\begin{theorem}
Let $\tilde{\boldsymbol{\Sigma}}(\rho)$ be given by \eqref{eq:tilde_sigma} wherein observations ${\bf x}_1,\cdots,{\bf x}_n$ are drawn according to Assumption \ref{ass:model}. Consider $\boldsymbol{\Sigma}_N={\bf U}\boldsymbol{\Lambda}_N{\bf U}^*$ the eigenvalue decomposition of $\boldsymbol{\Sigma}_N$. 
 Denote by ${\bf D}$ the diagonal matrix whose diagonal elements are solutions to the system of equations \eqref{eq:system}. Then, in the asymptotic large-$n$ regime, $\sqrt{n}\tilde{\boldsymbol{\delta}}=\sqrt{n}\left({\rm vec}(\tilde{\boldsymbol{\Sigma}}{(\rho)})-{\rm vec}(\boldsymbol{\Sigma}_0)\right)$ behaves as a zero-mean Gaussian distributed vector with covariance:
$$
\tilde{\bf M}_1=N^2(1-\rho)^2\left(\overline{\bf U}\boldsymbol{\Lambda}_N^{\frac{1}{2}}\otimes {\bf U}\boldsymbol{\Lambda}_N^{\frac{1}{2}}\right){\bf B}({\bf D})\left(\boldsymbol{\Lambda}_N^{\frac{1}{2}}{\bf U}^{\mbox{\tiny T}}\otimes \boldsymbol{\Lambda}_N^{\frac{1}{2}}{\bf U}^*\right)
$$
and pseudo-covariance: 
$$
\tilde{\bf M}_2=N^2(1-\rho)^2\left(\overline{\bf U}\boldsymbol{\Lambda}_N^{\frac{1}{2}}\otimes {\bf U}\boldsymbol{\Lambda}_N^{\frac{1}{2}}\right)\boldsymbol{\bf G}({\bf D})\left(\boldsymbol{\Lambda}_N^{\frac{1}{2}}{\bf U}^*\otimes \boldsymbol{\Lambda}_N^{\frac{1}{2}}{\bf U}^{\mbox{\tiny T}}\right).
$$
where ${\bf B}({\bf D})$ and ${\bf G}({\bf D})$ are given by \eqref{eq:B} and \eqref{eq:G} of Lemma~\ref{lemma:covariance}.
\label{th:clt_tilde}
\end{theorem}
Now that the fluctuations of $\tilde{\boldsymbol{\Sigma}}{(\rho)}$ have been determined, we are  in position to derive the asymptotic distribution of  ${\rm vec}(\hat{\bf C}_N(\rho))$. The very recent results in \cite{couillet-kammoun-14} establishing equality between the fluctuations of the bilinear-forms of $\hat{\bf C}_N(\rho)$ and those of its random equivalent $\hat{\bf S}_N(\rho)$ in the large-$n,N$ regime might lead us to expect similar results to hold in the large-$n$ regime.
As we will show in the following theorem, contrary to these first intuitions, the asymptotic distribution of ${\rm vec}(\tilde{\boldsymbol{\Sigma}}(\rho))$ is different from that of ${\rm vec}(\hat{\bf C}_N(\rho))$, even though it plays a central role in facilitating its analytical derivation. 
\begin{theorem}
Under the same setting of Theorem \ref{th:clt_tilde}, define $\tilde{\bf F}$ the $N^2\times N^2$ matrix:
$$
\tilde{\bf F}=N(1-\rho)\left(\overline{\bf U}{\bf D}^{\frac{1}{2}}\otimes {\bf U}{\bf D}^{\frac{1}{2}}\right)\tilde{\bf B}({\bf D})\left({\bf D}^{\frac{1}{2}}{\bf U}^{\mbox{\tiny T}}\otimes {\bf D}^{\frac{1}{2}}{\bf U}^*\right)
$$
with $\tilde{\bf B}({\bf D})$ defined in Lemma~\ref{lemma:covariance}.
Consider $\hat{\bf C}_N(\rho)$ the robust scatter estimator in \eqref{eq:hatc}. Then, in the large-$n$ asymptotic regime,  $\sqrt{n}\boldsymbol{\delta}=\sqrt{n}\left({\rm vec}(\hat{\bf C}_N(\rho))-{\rm vec}(\boldsymbol{\Sigma}_0)\right)$ behaves as a zero-mean Gaussian-distributed vector with covariance:
\begin{align*}
{\bf M}_1&=\left(\left(\boldsymbol{\Sigma}_0^{\frac{1}{2}}\right)^{\mbox{\tiny T}}\otimes \boldsymbol{\Sigma}_0^{\frac{1}{2}}\right)({\bf I}_{N^2}-\tilde{\bf F})^{-1}\left(\left(\boldsymbol{\Sigma}_0^{-\frac{1}{2}}\right)^{\mbox{\tiny T}}\otimes \boldsymbol{\Sigma}_0^{-\frac{1}{2}}\right)\tilde{\bf M}_1\\
&\times\left(\left(\boldsymbol{\Sigma}_0^{-\frac{1}{2}}\right)^{\mbox{\tiny T}}\otimes \boldsymbol{\Sigma}_0^{-\frac{1}{2}}\right)({\bf I}_{N^2}-\tilde{\bf F})^{-1}\left(\left(\boldsymbol{\Sigma}_0^{\frac{1}{2}}\right)^{\mbox{\tiny T}}\otimes \boldsymbol{\Sigma}_0^{\frac{1}{2}}\right)
\end{align*}
and pseudo-covariance:
\begin{align*}
{\bf M}_2&=\left(\left(\boldsymbol{\Sigma}_0^{\frac{1}{2}}\right)^{\mbox{\tiny T}}\otimes \boldsymbol{\Sigma}_0^{\frac{1}{2}}\right)({\bf I}_{N^2}-\tilde{\bf F})^{-1}\left(\left(\boldsymbol{\Sigma}_0^{-\frac{1}{2}}\right)^{\mbox{\tiny T}}\otimes \boldsymbol{\Sigma}_0^{-\frac{1}{2}}\right)\tilde{\bf M}_2\\
&\times\left(\boldsymbol{\Sigma}_0^{-\frac{1}{2}}\otimes \left(\boldsymbol{\Sigma}_0^{-\frac{1}{2}}\right)^{\mbox{\tiny T}}\right)({\bf I}_{N^2}-\tilde{\bf F}^{\mbox{\tiny T}})^{-1}\left(\boldsymbol{\Sigma}_0^{\frac{1}{2}}\otimes \left(\boldsymbol{\Sigma}_0^{\frac{1}{2}}\right)^{\mbox{\tiny T}}\right).
\end{align*}
\label{th:clt}
\end{theorem}
\begin{proof}
The proof is deferred to Appendix \ref{app:clt}
\end{proof}

\section{Numerical results}
In all our simulations, we consider the case where ${\bf x}_1,\cdots,{\bf x}_n$ are {independent zero-mean} Gaussian random vectors with covariance matrix $\boldsymbol{\Sigma}_N$ of Toeplitz form:
\begin{equation}
\left[{\bf C}_N\right]_{i,j}=\left\{\begin{array}{ll}b^{j-i} &\hspace{0.1cm} i\leq j\\
\left(b^{i-j}\right)^*  &\hspace{0.1cm}i>j
\end{array}
, \hspace{0.9cm}|b|\in\left]0,1\right[,
\right.
\label{eq:CN}
\end{equation}
\subsection{Which regime is expected to be more accurate}
In order to study the {behavior} of RTE, assumptions letting the number of observations and/or their sizes increase to infinity are essential for tractability. 
The {behavior} of RTE is studied under both concurrent asymptotic regimes, namely the large-$n$ regime, which underlies all the derivations  of this paper, and  the $n,N$-large regime recently considered in \cite{couillet-kammoun-14}.
Given that the scope of the results derived in the {large-$n,N$} regime, has thus far been limited to the  handling of bilinear forms, practitioners might wonder to know whether, for their specific scenario, further investigation of this regime would produce more accurate results. In this first experiment, we attempt to answer to this open question by noticing that both regimes have the common denominator of producing random matrices that act as equivalents to the robust-scatter estimator. The accuracy of each regime is thus evaluated by measuring the closeness of the robust-scatter estimator to its random equivalent proposed by each regime. This closeness is measured using the following metrics:
$$
\mathcal{E}_{n}\triangleq\frac{1}{N}\mathbb{E}\left\|\hat{\bf C}(\rho)-\tilde{\boldsymbol{\Sigma}}(\rho)\right\|_{{\rm Fro}}^2
$$
and 
$$
\mathcal{E}_{n,N}\triangleq \frac{1}{N}\mathbb{E}\left\|\hat{\bf C}(\rho)-\hat{\bf S}_N(\rho)\right\|_{{\rm Fro}}^2.
$$
Figures \ref{fig:N4}, \ref{fig:N16} and \ref{fig:N32} represent these metrics with respect to the ratio $\frac{n}{N}$ when $N=4,16,32$, $b=0.7$ and $\rho$ set to $0.5$. The region over which the use of the large-$n$ regime is  recommended corresponds to the values of $\frac{n}{N}$ for which the $\mathcal{E}_n$ curve  is below the $\mathcal{E}_{n,N}$ one.  

From these figures, it appears that, as $N$ increases, the region over which results derived under the large-$n$ regime are more accurate, corresponds to  larger values of the ratio $\frac{n}{N}$.
\begin{figure}
  \begin{center}
\includegraphics[scale=1]{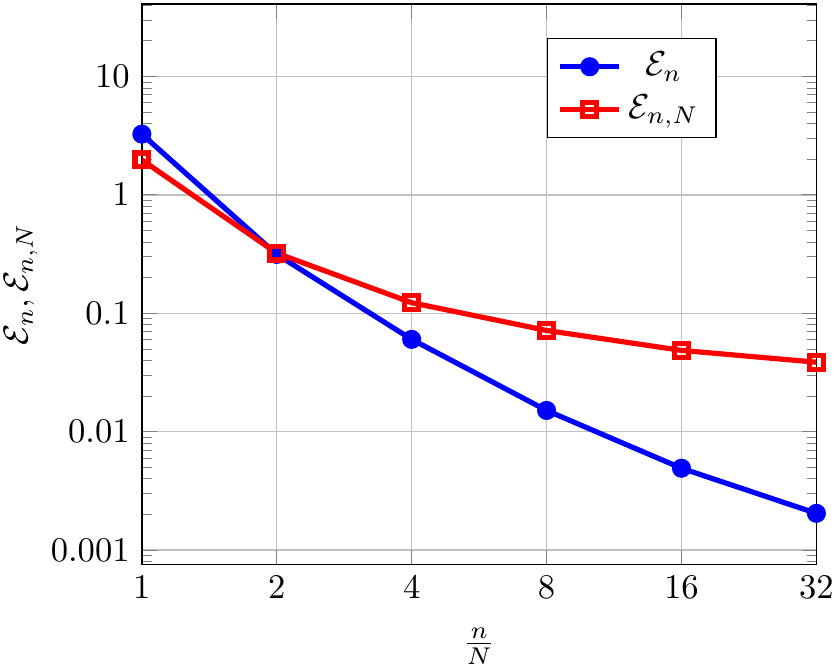}
\caption{Accuracy of the random equivalent when $N=4$}
\label{fig:N4}
\end{center}
\end{figure}

\begin{figure}
  \begin{center}
\includegraphics[scale=1]{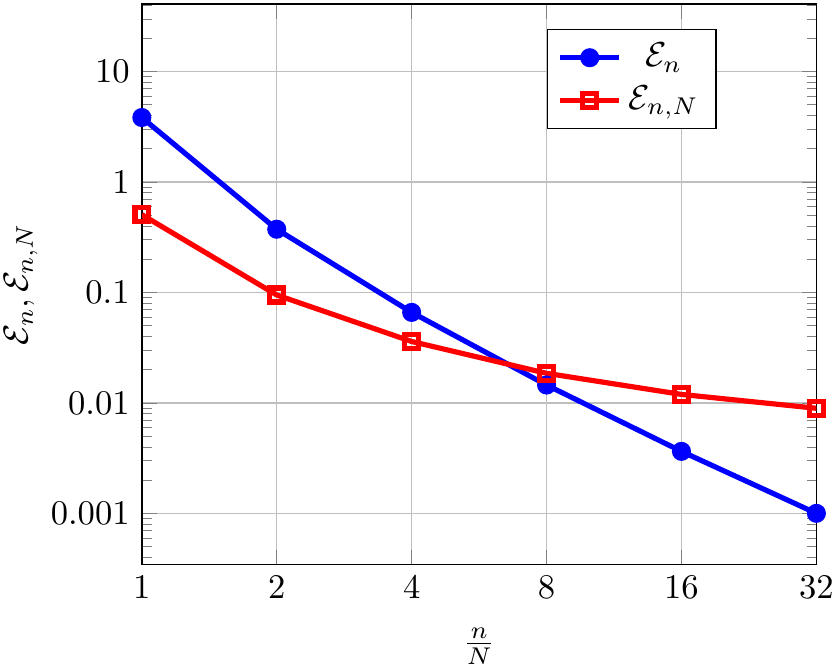}
\caption{Accuracy of the random equivalent when $N=16$}
\label{fig:N16}
\end{center}
\end{figure}

\begin{figure}
  \begin{center}
\includegraphics[scale=1]{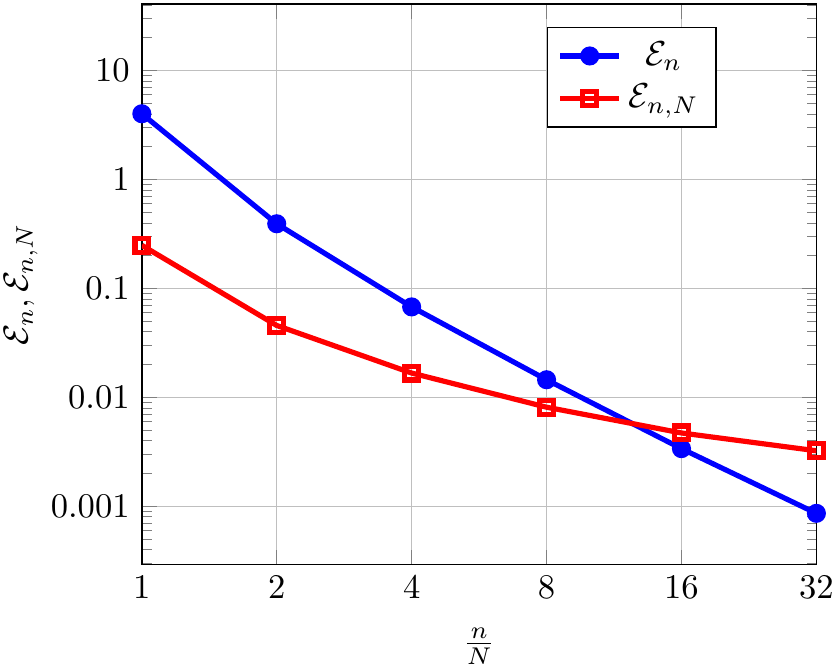}
\caption{Accuracy of the random equivalent when $N=32$}
\label{fig:N32}
\end{center}
\end{figure}


\subsection{Asymptotic bias}
In this section, we assess the  bias of the RTE with respect to the population covariance matrix. Since in many applications in radar detection, we only need to estimate the covariance matrix up to a scale factor, we define the bias as:
$$
{\rm Bias}=\left\|\mathbb{E}\left[\frac{N}{\tr\left(\boldsymbol{\Sigma}_N^{-1}\hat{\bf C}_N\right)}\boldsymbol{\Sigma}_N^{-1}\hat{\bf C}_N\right]-{\bf I}_N\right\|_{\rm Fro}^2.
$$
Since $\frac{N}{\tr\left(\boldsymbol{\Sigma}_N^{-1}\hat{\bf C}_N\right)}\boldsymbol{\Sigma}_N^{-1}\hat{\bf C}_N$ has a bounded spectral norm, the dominated convergence theorem implies that:
$$
{\rm Bias}\xrightarrow[n\to+\infty]{} \left\|\left[\frac{N}{\tr\left(\boldsymbol{\Sigma}_N^{-1}\boldsymbol{\Sigma}_0\right)}\boldsymbol{\Sigma}_N^{-1}\boldsymbol{\Sigma}_0\right]-{\bf I}_N\right\|_{\rm Fro}^2.
$$
Figure \ref{fig:bias} displays the asymptotic and empirical bias with respect to the Toeplitz coefficient $b$ and for $\rho=0.2,0.5,0.9$. We note that the bias is an increasing function of $b$. This is expected since for small values of $b$, the covariance matrix becomes close to the identity matrix. The RTE, viewed as a shrunk version of the Tyler to the identity matrix will thus produce small values of bias.
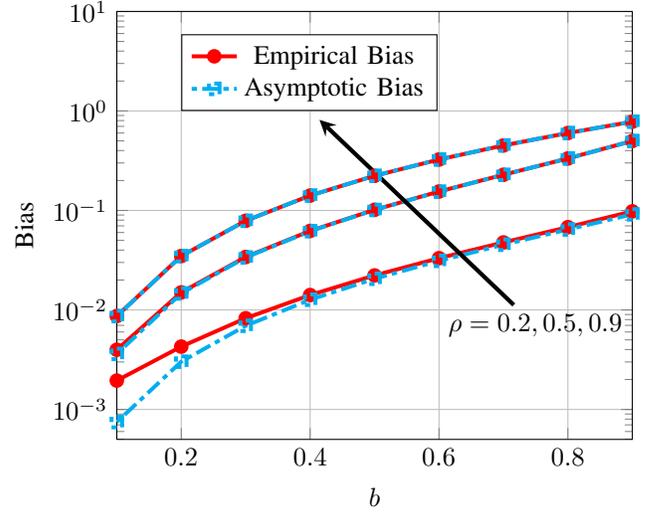
\begin{figure}
  \begin{center}
   \begin{tikzpicture}[scale=1,font=\normalsize]
\tikzset{>=latex}
\tikzset{dashdot/.style={dash pattern=on 2pt off 2pt on 6pt off 2pt}}
\tikzstyle{dotted}=[dash pattern=on \pgflinewidth off 2pt]
     \begin{axis}[
      ymode=log,
      xmin=0.1,
      ymin=0.0005,
      xmax=0.9,
      ymax=10,
      grid=major,
      legend style={ at={(0,0)},
      anchor= south west,
      at={(axis cs:0.2,1)}},
      xlabel={$b$},
      ylabel={Bias},
      ]
\addplot[color=red,line width=1.5pt,mark=*] coordinates{
(0.100000,0.001950)(0.200000,0.004280)(0.300000,0.008254)(0.400000,0.014116)(0.500000,0.022235)(0.600000,0.033138)(0.700000,0.047700)(0.800000,0.067934)(0.900000,0.097917)};
\addlegendentry{Empirical Bias};
\addplot[color=cyan,line width=1.5pt,mark=square,dashdot] coordinates{
(0.100000,0.000747)(0.200000,0.003021)(0.300000,0.006930)(0.400000,0.012674)(0.500000,0.020576)(0.600000,0.031155)(0.700000,0.045256)(0.800000,0.064352)(0.900000,0.091315)
};
\addlegendentry{Asymptotic Bias};
\addplot[color=red,line width=1.5pt,mark=*] coordinates{
(0.100000,0.003980)(0.200000,0.015009)(0.300000,0.034080)(0.400000,0.062252)(0.500000,0.101473)(0.600000,0.154927)(0.700000,0.228298)(0.800000,0.333067)(0.900000,0.499290)
};
\addplot[color=cyan,line width=1.5pt,mark=square,dashdot] coordinates{
(0.100000,0.003615)(0.200000,0.014649)(0.300000,0.033698)(0.400000,0.061882)(0.500000,0.101097)(0.600000,0.154531)(0.700000,0.227873)(0.800000,0.332593)(0.900000,0.498636)
};
\addplot[color=red,line width=1.5pt,mark=*] coordinates{
(0.100000,0.008686)(0.200000,0.034828)(0.300000,0.078776)(0.400000,0.141124)(0.500000,0.222765)(0.600000,0.325025)(0.700000,0.449787)(0.800000,0.599873)(0.900000,0.779918)
};
\addplot[color=cyan,line width=1.5pt,mark=square,dashdot] coordinates{
(0.100000,0.008677)(0.200000,0.034819)(0.300000,0.078765)(0.400000,0.141109)(0.500000,0.222755)(0.600000,0.325014)(0.700000,0.449776)(0.800000,0.599866)(0.900000,0.779911)
};
\node[anchor=west] (source) at (axis cs:0.6,0.007){%
       $\rho=0.2,0.5,0.9$
       };
\node (destination) at (axis cs:0.4,1){};
       \draw[->,>=stealth,ultra thick](source)--(destination);
\end{axis}
\end{tikzpicture}
\caption{Analysis of the bias for $n=1000$ and $N=2$ with respect to $b$ and for different values of $\rho$.}
\label{fig:bias}
\end{center}
\end{figure}
\subsection{Central Limit Theorem}
The central limit theorem provided in this paper can help determine fluctuations of any continuous functional of ${\rm vec}(\hat{\bf C}_N)$. As an application, we consider in this section the quadratic form of type $\frac{1}{N}{\bf p}^{*}\hat{\bf C}_N^{-1}(\rho){\bf p}$ with $\|{\bf p}\|=1$ (used for instance for detection in array processing problems \cite{vtree2002oap}),  for which the large-$n$ and the large-$n,N$ regimes predict different kind of fluctuations. 
As a matter of fact, applying the Delta Method \cite{vaart}, one can easily prove that under the large-$n$,
\begin{align*}
&{T}_n\triangleq\frac{\sqrt{n}\left(\frac{1}{N}{\bf p}^{*}\hat{\bf C}_N^{-1}(\rho){\bf p}-\frac{1}{N}{\bf p}^{*}\boldsymbol{\Sigma}_0^{-1}(\rho){\bf p}\right)}{\sqrt{\frac{1}{N^2}\left((\boldsymbol{\Sigma}_0^{-1})^{\mbox{\tiny T}}\overline{\bf p}\otimes \boldsymbol{\Sigma}_0^{-1}{\bf p}\right)^{*}{\bf M}_1\left((\boldsymbol{\Sigma}_0^{-1})^{\mbox{ \tiny T}}\overline{\bf p}\otimes \boldsymbol{\Sigma}_0^{-1}{\bf p}\right)}}\\
&\xrightarrow[]{\mathcal{D}}\mathcal{N}(0,1).
\end{align*}
On the other hand, using results from \cite{couillet-kammoun-14}, {one} can prove that under the large-$n,N$ regime, $\frac{\sqrt{n}}{N}{\bf p}^{*}\hat{\bf C}_N^{-1}(\rho){\bf p}$ {satisfies:}
\begin{align*}
&T_{n,N}\triangleq\sqrt{\frac{n}{\sigma_N^2}}\left(\frac{1}{N}{\bf p}^{*}\hat{\bf C}_N^{-1}(\rho){\bf p}-\frac{1}{N}{\bf p}^{*}{\bf Q}_N(\underline{\rho}){\bf p}\right)\xrightarrow[]{\mathcal{D}}\mathcal{N}(0,1)
\end{align*}
where:
$$
\sigma_N^2=\frac{m(-\underline{\rho})^2(1-\underline{\rho})^2\left(\frac{1}{N}{\bf p}^*\boldsymbol{\Sigma}_N{\bf Q}_N^2{\bf p}\right)^2}{{\rho}^2(1-cm(-\underline{\rho})^2(1-\underline{\rho})^2\frac{1}{N}\boldsymbol{\Sigma}_N^2{\bf Q}_N^2(\underline{\rho}))}
$$
with $\underline{\rho}$, $m(-\underline{\rho})$ and ${\bf Q}(\underline{\rho})$ have the same expressions as in \cite{couillet-kammoun-14} when ${\bf C}_N$ in \cite{couillet-kammoun-14} is replaced by $\boldsymbol{\Sigma}_N$. A natural question that arises is  which of the two competing results is the most reliable for a particular set of values $N$ and $n$. To answer this question, we plot in figures \ref{fig:clt_4}, \ref{fig:clt_16} and \ref{fig:clt_32} the Kolmogorov-Smirnov distance,  between the empirical distribution function of $T_n$ and $T_{n,N}$ obtained over $50\,000$ realizations,  and the standard normal distribution with respect to the ratio $\frac{n}{N}$ when $b=0.7\jmath, \rho=0.5, {\bf p}=\left[1,\cdots,1\right]$ and for $N=4,16,32$. We note that for values of $N$  up to $16$, results derived under the large-$n$ regime are more accurate for a large range of $n$ while the use of the results from the large-$n,N$ regime seems to be recommended for $N=32$. 
  
\begin{figure}
  \begin{center}
\includegraphics[scale=1]{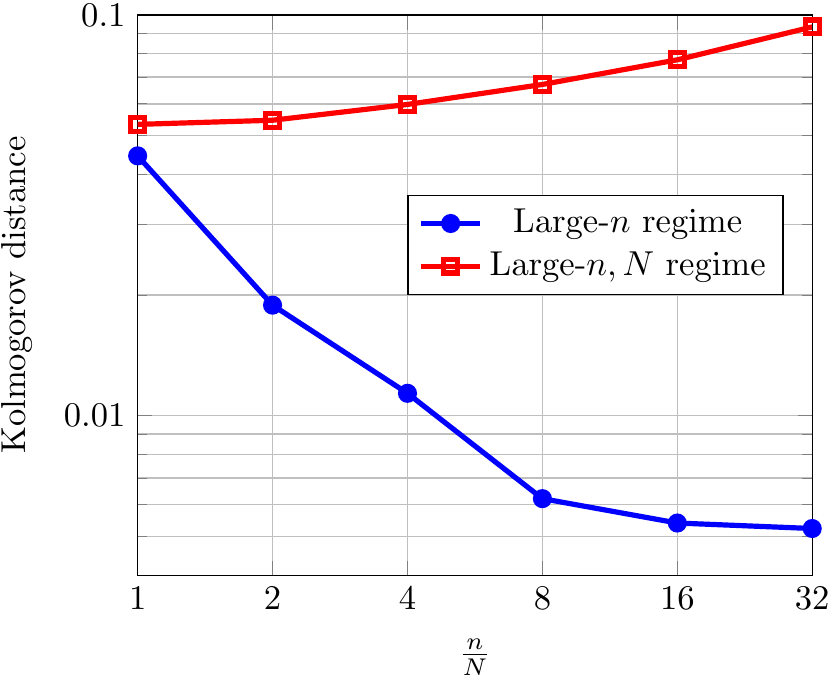}
\caption{Analysis of the accuracy of the CLT results for $N=4$}
\label{fig:clt_4}
\end{center}
\end{figure}

\begin{figure}
  \begin{center}
\includegraphics[scale=1]{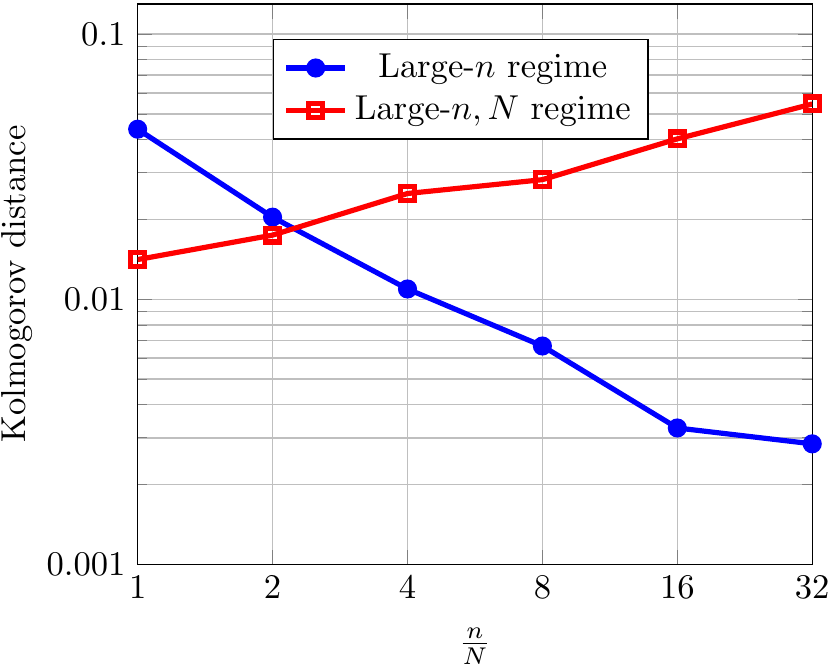}
\caption{Analysis of the accuracy of the CLT results for $N=16$}
\label{fig:clt_16}
\end{center}
\end{figure}

\begin{figure}
  \begin{center}
\includegraphics[scale=1]{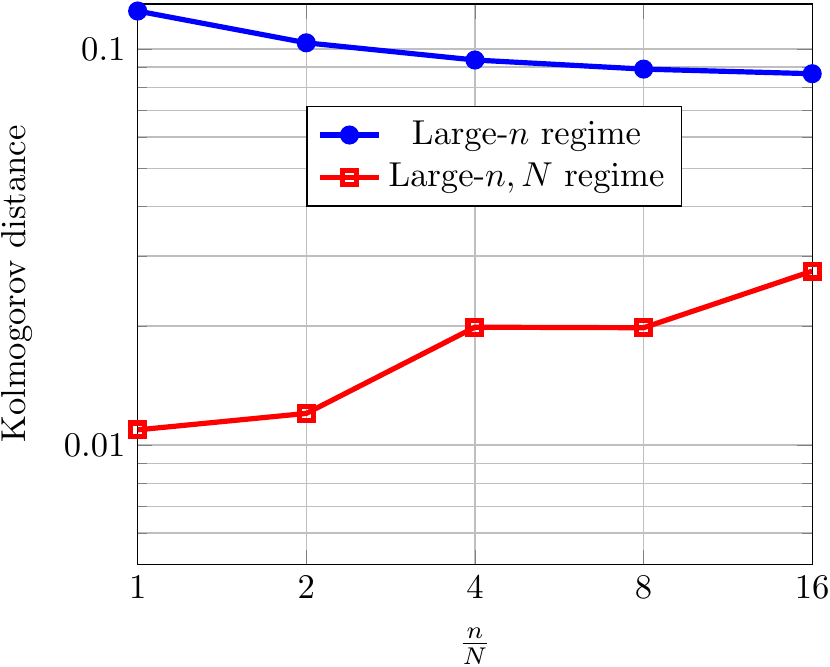}
\caption{Analysis of the accuracy of the CLT results for $N=32$}
\label{fig:clt_32}
\end{center}
\end{figure}

\section{Conclusions}
This paper focuses on the  statistical behavior of the RTE. 
It is worth noticing that despite the popularity of the RTE, characterizing its statistical properties  has remained unclear until the work in \cite{couillet-kammoun-14} shedding light on its behavior when the large-$n,N$ regime is considered (the number of observations $n$ and their size $N$ growing simultaneously to infinity.).  Interestingly, no results were provided for the standard large-$n$ regime in which $N$ is fixed while $n$ goes to infinity. This has motivated our work. 
 In particular, we established in this paper  that the RTE converges, under the large-$n$ regime, to a deterministic matrix which differs as expected from the true population covariance matrix. An important feature of this results is that it allows for the computation of the asymptotic bias incurred by the use of the RTE. 
 We also studied the fluctuations of the RTE around its limit and prove that they converge to a multivariate Gaussian distribution with zero mean and a covariance matrix depending on the true population covariance and the regularization parameter.  The characterization of these fluctuations are paramount to applications of radar detection in which  RTEs are used. 
 Finally, numerical simulations were carried out in order to validate the theoretical results and also to assess their accuracy with their counterparts obtained under the large-$n,N$ regime.  

\appendices

\section{Proof of Lemma \ref{lemma:bounded_spectral}}
\label{app:bounded_spectral}
{In the following appendices, for readability purposes, the notation $\boldsymbol{\Sigma}_0(\rho)$ (resp. $\tilde{\boldsymbol{\Sigma}}(\rho)$)  is simply replaced by $\boldsymbol{\Sigma}_0$ (resp. $\tilde{\boldsymbol{\Sigma}}$). Of course, the dependence of $\boldsymbol{\Sigma}_0$ to $\rho$ is not omitted.}

Multiplying both sides of \eqref{eq:sigma_0} by $\boldsymbol{\Sigma}_N^{-1}$, we show that $\boldsymbol{\Sigma}_0$ satisfies:
$$
(1-\rho)\mathbb{E}\left[\frac{{\bf w}{\bf w}^*}{\frac{1}{N}{\bf w}^*\boldsymbol{\Sigma}_N^{\frac{1}{2}}\boldsymbol{\Sigma}_0^{-1}\boldsymbol{\Sigma}_N^{\frac{1}{2}}{\bf w}}\right]+\rho \boldsymbol{\Sigma}_N^{-1}=\boldsymbol{\Sigma}_N^{-\frac{1}{2}}\boldsymbol{\Sigma}_0\boldsymbol{\Sigma}_N^{-\frac{1}{2}},
$$ 
where ${\bf w}$ is zero-mean distributed with covariance {matrix} ${\bf I}_N$. Define ${\bf A}=\boldsymbol{\Sigma}_N^{-\frac{1}{2}}\boldsymbol{\Sigma}_0\boldsymbol{\Sigma}_N^{-\frac{1}{2}}$. Then,
$$
{\bf A}=(1-\rho)\mathbb{E}\left[\frac{{\bf w}{\bf w}^*}{\frac{1}{N}{\bf w}^*{\bf A}^{-1}{\bf w}}\right]+\rho\boldsymbol{\Sigma}^{-1}
$$
which yields the following bound for $\|{\bf A}\|$,
$$
\|{\bf A}\|\leq (1-\rho)\|{\bf A}\|+\frac{\rho}{\lambda_{\rm min}(\boldsymbol{\Sigma}_N)}.
$$
Hence,
\begin{equation}
\|{\bf A}\| \leq \frac{1}{\lambda_{\rm min}(\boldsymbol{\Sigma}_N)}.
\label{eq:upper_bound}
\end{equation}
Now, $\|{\bf A}\|$ can be lower-bounded by:
\begin{align}
\|{\bf A}\|&=\max_{\|{\bf x}\|=1} {\bf x}^*\boldsymbol{\Sigma}_N^{-\frac{1}{2}}\boldsymbol{\Sigma}_0\boldsymbol{\Sigma}_N^{-\frac{1}{2}}{\bf x} \nonumber\\
&\stackrel{(a)}{\geq} \|\boldsymbol{\Sigma}_0\| \max_{\|{\bf x}\|=1} {\bf x}^*\boldsymbol{\Sigma}_N^{-\frac{1}{2}}{\bf u}{\bf u}^*\boldsymbol{\Sigma}_N^{-\frac{1}{2}}{\bf x}\nonumber\\
&\geq \|\boldsymbol{\Sigma}_0\| {\bf u}^* \boldsymbol{\Sigma}_N^{-\frac{1}{2}}{\bf u}{\bf u}^*\boldsymbol{\Sigma}_N^{-\frac{1}{2}}{\bf u} \nonumber\\
&\geq \frac{\|\boldsymbol{\Sigma}_0\|}{\|\boldsymbol{\Sigma}_N\|}, \label{eq:lower_bound}
\end{align}
where in $(a)$ ${\bf u}$ is the eigenvector corresponding to the maximum eigenvalue of $\boldsymbol{\Sigma}_0$.
Combining \eqref{eq:upper_bound} and \eqref{eq:lower_bound}, we thus obtain:
$$
\|{\boldsymbol{\Sigma}_0}\|\leq \frac{\|\boldsymbol{\Sigma}_N\|}{\lambda_{\rm min}(\boldsymbol{\Sigma}_N)}.
$$

\section{Proof of Theorem \ref{th:first_order}}
\label{app:first_order}
The proof is based on controlling the random elements $d_i(\rho)$ given by:
$$
d_i(\rho)=\frac{{\bf x}_i^*\hat{\bf C}_N^{-1}(\rho){\bf x}_i-{\bf x}_i^*{\boldsymbol{\Sigma}_0^{-1}}{\bf x}_i}{\sqrt{{\bf x}_i^*{\boldsymbol{\Sigma}_0^{-1}}{\bf x}_i}\sqrt{{\bf x}_i^*\hat{\bf C}_N^{-1}(\rho){\bf x}_i}}.
$$
Recall that, by the {SLLN}, under the large-$n$ regime,  ${\boldsymbol{\Sigma}_0}$ satisfies:
$$
{\boldsymbol{\Sigma}_0}=N(1-\rho)\frac{1}{n}\sum_{i=1}^n \frac{{\bf x}_i{\bf x}_i^*}{{\bf x}_i^*{\boldsymbol{\Sigma}_0^{-1}}{\bf x}_i}+\rho{\bf I}_N+\boldsymbol{\epsilon}_n(\rho),
$$
where $\boldsymbol{\epsilon}_n$ is a $N\times N$ matrix whose elements converge almost surely to zero and satisfy $\left[\boldsymbol{\epsilon}_n(\rho)\right]_{i,j}=\mathcal{O}_p(\frac{1}{n})$.

In the sequel, we prove that for any $\kappa>0$, 
$$
\sup_{\rho\in\left[\kappa,1\right]}\max_{1\leq i\leq n}|d_i(\rho)|\asto 0.
$$
For that, we need to work out the differences ${\bf x}_i^*\hat{\bf C}_N^{-1}(\rho){\bf x}_i-{\bf x}_i^*{\boldsymbol{\Sigma}_0^{-1}}{\bf x}_i$ for $i=1,\cdots,n$. Using the resolvent identity ${\bf A}^{-1}-{\bf B}^{-1}={\bf A}^{-1}\left({\bf B}-{\bf A}\right){\bf B}^{-1}$ for any $N\times N$ invertible matrices, we obtain:
\begin{align*}
&{\bf x}_j^*\hat{\bf C}_N^{-1}(\rho){\bf x}_j-{\bf x}_j^*{\boldsymbol{\Sigma}_0^{-1}}{\bf x}_j\\
&={\bf x}_j^*\hat{\bf C}_N^{-1}\left[\frac{1-\rho}{n}\sum_{i=1}^n\frac{{\bf x}_i{\bf x}_i^*\left(\frac{1}{N}{\bf x}_i^*\boldsymbol{\Sigma}_0^{-1}{\bf x}_i-\frac{1}{N}{\bf x}_i^*\hat{\bf C}_N^{-1}(\rho){\bf x}_i\right)}{\frac{1}{N}{\bf x}_i^*{\boldsymbol{\Sigma}_0^{-1}}{\bf x}_i \frac{1}{N}{\bf x}_i^*\hat{\bf C}_N^{-1}(\rho){\bf x}_i}\right.\\
&\left.+\boldsymbol{\epsilon}_n\right.\Bigg] \boldsymbol{\Sigma}_0^{-1}{\bf x}_j\\
&=\frac{1-\rho}{n}\sum_{i=1}^n \frac{{\bf x}_j^*\hat{\bf C}_N^{-1}(\rho){\bf x}_i{\bf x}_i^*{\boldsymbol{\Sigma}_0^{-1}}{\bf x}_jd_i(\rho)}{\sqrt{\frac{1}{N}{\bf x}_i^*\hat{\bf C}_N^{-1}(\rho){\bf x}_i\frac{1}{N}{\bf x}_i^*{\boldsymbol{\Sigma}_0^{-1}}{\bf x}_i}}\\
&+{\bf x}_j^*\hat{\bf C}_N^{-1}(\rho)\boldsymbol{\epsilon}_n{\boldsymbol{\Sigma}_0^{-1}}{\bf x}_j.
\end{align*}
Hence, 
\begin{align*}
d_j(\rho)&=\frac{\frac{1-\rho}{n}\sum_{i=1}^n \frac{{\bf x}_j^*\hat{\bf C}_N^{-1}(\rho){\bf x}_i{\bf x}_i^*{\boldsymbol{\Sigma}_0^{-1}}{\bf x}_j d_i(\rho)}{\sqrt{\frac{1}{N}{\bf x}_i^*{\boldsymbol{\Sigma}_0^{-1}}{\bf x}_i\frac{1}{N}{\bf x}_i^*\hat{\bf C}_N^{-1}(\rho){\bf x}_i}}}{\sqrt{{\bf x}_j^*\hat{\bf C}_N^{-1}(\rho){\bf x}_j{\bf x}_j^*{\boldsymbol{\Sigma}_0^{-1}}{\bf x}_j }}\\
&+\frac{{\bf x}_j^*\hat{\bf C}_N^{-1}(\rho)\boldsymbol{\epsilon}_n{\boldsymbol{\Sigma}_0^{-1}}{\bf x}_j}{\sqrt{{\bf x}_j^*\hat{\bf C}_N^{-1}(\rho){\bf x}_j{\bf x}_j^*{\boldsymbol{\Sigma}_0^{-1}}{\bf x}_j }}.
\end{align*}
Let $d_{\rm max}(\rho)=\max_{1\leq j\leq n}|d_j(\rho)|$. By the Cauchy-Schwartz inequality, we thus obtain:
\begin{align*}
d_{\rm max}(\rho)&\leq \frac{d_{\rm max}(\rho)}{\sqrt{{\bf x}_j^*\hat{\bf C}_N^{-1}(\rho){\bf x}_j{\bf x}_j^*{\boldsymbol{\Sigma}_0^{-1}}{\bf x}_j}}\\
&\times \sqrt{\frac{1-\rho}{n}\sum_{i=1}^n\frac{{\bf x}_j^*\hat{\bf C}_N^{-1}(\rho){\bf x}_i{\bf x}_i^*\hat{\bf C}_N^{-1}(\rho){\bf x}_j}{\frac{1}{N}{\bf x}_i^*\hat{\bf C}_N^{-1}(\rho){\bf x}_i}}\\
&\times\sqrt{\frac{1-\rho}{n}\sum_{i=1}^n\frac{{\bf x}_j^*{\boldsymbol{\Sigma}_0^{-1}}{\bf x}_i{\bf x}_i^*{\boldsymbol{\Sigma}_0^{-1}}{\bf x}_j}{\frac{1}{N}{\bf x}_i^*{\boldsymbol{\Sigma}_0^{-1}}{\bf x}_i}}\\
&+\|\hat{\bf C}_N^{-\frac{1}{2}}(\rho)\boldsymbol{\epsilon}_n\boldsymbol{\Sigma}_0^{-\frac{1}{2}}\|.
\end{align*}
Therefore,
\begin{align*}
d_{\rm max}(\rho)&\leq \frac{d_{\rm max}(\rho)}{\sqrt{{\bf x}_j^*\hat{\bf C}_N^{-1}(\rho){\bf x}_j{\bf x}_j^*{\boldsymbol{\Sigma}_0^{-1}}{\bf x}_j}}\\
&\times \sqrt{{\bf x}_j^*\hat{\bf C}_N^{-\frac{1}{2}}\left({\bf I}_N-\rho\hat{\bf C}_N^{-1}(\rho)\right)\hat{\bf C}_N^{-\frac{1}{2}}{\bf x}_j}\\
&\times\sqrt{{\bf x}_j^*\boldsymbol{\Sigma}_0^{-\frac{1}{2}}\left({\bf I}_N-\rho{\boldsymbol{\Sigma}_0^{-1}}\right)\boldsymbol{\Sigma}_0^{-\frac{1}{2}}{\bf x}_j-{\bf x}_j^*\boldsymbol{\Sigma}_0^{-1}\boldsymbol{\epsilon}_n\boldsymbol{\Sigma}_0^{-1}{\bf x}_j}\\
&+\|\hat{\bf C}_N^{-\frac{1}{2}}(\rho)\boldsymbol{\epsilon}_n\boldsymbol{\Sigma}_0^{-\frac{1}{2}}\|.
\end{align*}
Using the relation $\left|{\bf x}^*{\bf A}{\bf y}\right|\leq \|{\bf x}\|\|{\bf A}\|\|{\bf y}\|$, we thus obtain:
\begin{align*}
&d_{\rm max}(\rho)\leq d_{\rm max}(\rho) \sqrt{\|{\bf I}_N-\rho\hat{\bf C}_N^{-1}(\rho)\|}\\
&\left(\|{\bf I}_N-\rho{\boldsymbol{\Sigma}_0^{-1}}\|-\frac{{\bf x}_j^*\boldsymbol{\Sigma}_0^{-1}\boldsymbol{\epsilon}_n\boldsymbol{\Sigma}_0^{-1}{\bf x}_j}{{\bf x}_j^*\boldsymbol{\Sigma}_0^{-1}{\bf x}_j}\right)^{\frac{1}{2}}+\|\hat{\bf C}_N^{-\frac{1}{2}}(\rho)\boldsymbol{\epsilon}_n\boldsymbol{\Sigma}_0^{-\frac{1}{2}}\|.
\end{align*}
Since  $\sup_{\rho\in\left[\kappa,1\right)}\|\hat{\bf C}_N^{-\frac{1}{2}}\boldsymbol{\epsilon}_n(\rho)\boldsymbol{\Sigma}_0^{-\frac{1}{2}}\|\leq \frac{1}{\kappa}\sup_{\rho\in\left[\kappa,1\right)}\|\boldsymbol{\epsilon}_n(\rho)\|$ and using the fact that $\|{\bf I}_N-\rho\hat{\bf C}_N^{-1}(\rho)\|\leq 1$, we get:
\begin{align*}
d_{\rm max}(\rho) &\leq d_{\rm max}(\rho)\left(\sqrt{\|{\bf I}_N-\rho{\boldsymbol{\Sigma}_0^{-1}}\|}\right.\\
&\left.+\sqrt{\|\boldsymbol{\Sigma}_0^{-\frac{1}{2}}\boldsymbol{\epsilon}_n\boldsymbol{\Sigma}_0^{-\frac{1}{2}}\|}\right) +\frac{1}{\kappa}\|\boldsymbol{\epsilon}_n\|.
\end{align*}
Again, as $\|\boldsymbol{\Sigma}_0^{-\frac{1}{2}}\boldsymbol{\epsilon}_n\boldsymbol{\Sigma}_0^{-\frac{1}{2}}\| \leq \frac{\left\|\boldsymbol{\epsilon}_n\right\|}{\kappa}$,  we have:
$$
d_{\rm max}(\rho)\left(1-\sqrt{\|{\bf I}_N-\rho{\boldsymbol{\Sigma}_0^{-1}}\|}-\sqrt{\frac{1}{\kappa}\|\boldsymbol{\epsilon}_n\|}\right) \leq \frac{1}{\kappa}\|\boldsymbol{\epsilon}_n\|.
$$
From Lemma \ref{lemma:bounded_spectral}, $\left\|{\boldsymbol{\Sigma}_0}\right\| \leq \frac{\|\boldsymbol{\Sigma}_N\|}{\lambda_{\rm min}(\boldsymbol{\Sigma}_N)}$. Therefore, for $n$ large enough (say large enough for the left-hand parenthesis to be greater than zero),
$$
d_{\rm max}(\rho) \leq \frac{\frac{1}{\kappa}\|\boldsymbol{\epsilon}_n\|}{1-\sqrt{1-\rho\frac{\lambda_{\rm min}(\boldsymbol{\Sigma}_N)}{\|\boldsymbol{\Sigma}_N\|}}-\sqrt{\frac{1}{\kappa}\|\boldsymbol{\epsilon}_n\|}}.
$$
Taking the supremum over $\rho\in\left[\kappa,1\right)$, we finally obtain:
$$
\sup_{\rho\in\left[\kappa,1\right)} d_{\rm max}(\rho) \leq \frac{\frac{1}{\kappa}\|\boldsymbol{\epsilon}_n\|}{1-\sqrt{1-\kappa\frac{\lambda_{\rm min}(\boldsymbol{\Sigma}_N)}{\|\boldsymbol{\Sigma}_N\|}}-\sqrt{\frac{1}{\kappa}\|\boldsymbol{\epsilon}_n\|}}.
$$
thereby showing that $d_{\rm max}(\rho)\asto 0$ and $d_{\rm max}(\rho)=\mathcal{O}_p\left(\frac{1}{n}\right)$
Now, that the control of $d_{\rm max}(\rho)$ is performed, we are in position to handle the difference $\hat{\bf C}_N(\rho)-{\boldsymbol{\Sigma}_0}$. We have:
\begin{align*}
\hat{\bf C}_N(\rho)-{\boldsymbol{\Sigma}_0}&=\frac{1-\rho}{n}\sum_{i=1}^n \frac{{\bf x}_i{\bf x}_i^*\left({\bf x}_i^*{\boldsymbol{\Sigma}_0^{-1}}{\bf x}_i-{\bf x}_i^*\hat{\bf C}_N^{-1}(\rho){\bf x}_i\right)}{{\bf x}_i^*\hat{\bf C}_N^{-1}(\rho){\bf x}_i\frac{1}{N}{\bf x}_i^*{\boldsymbol{\Sigma}_0^{-1}}{\bf x}_i}\\
&-\boldsymbol{\epsilon}_n(\rho)\\
&=\frac{1-\rho}{n}\sum_{i=1}^n \frac{-{\bf x}_i{\bf x}_i^*d_i(\rho)}{\sqrt{\frac{1}{N}{\bf x}_i^*\hat{\bf C}_N^{-1}(\rho){\bf x}_i}\sqrt{\frac{1}{N}{\bf x}_i^*{\boldsymbol{\Sigma}_0^{-1}}{\bf x}_i}}\\
&-\boldsymbol{\epsilon}_n(\rho).
\end{align*}
Therefore,
\begin{align*}
&\|\hat{\bf C}_N(\rho)-{\boldsymbol{\Sigma}_0}\|\\
&\leq d_{\rm max}(\rho)\left\|\frac{1-\rho}{n}\sum_{i=1}^n\frac{{\bf x}_i{\bf x}_i^*}{\sqrt{\frac{1}{N}{\bf x}_i^*\hat{\bf C}_N^{-1}(\rho){\bf x}_i}\sqrt{\frac{1}{N}{\bf x}_i^*{\boldsymbol{\Sigma}_0^{-1}}{\bf x}_i}}\right\|\\ 
&+\left\|\boldsymbol{\epsilon}_n(\rho)\right\|.
\end{align*}
By the Cauchy-Schwartz inequality, we get:
\begin{align*}
\|\hat{\bf C}_N(\rho)-{\boldsymbol{\Sigma}_0}\|&\leq  d_{\rm max}(\rho) \left\|\frac{1-\rho}{n}\sum_{i=1}^n \frac{{\bf x}_i{\bf x}_i^*}{\frac{1}{N}{\bf x}_i^*\hat{\bf C}_N^{-1}(\rho){\bf x}_i}\right\|^{\frac{1}{2}}\\
&\times \left\|\frac{1-\rho}{n}\sum_{i=1}^n \frac{{\bf x}_i{\bf x}_i^*}{\frac{1}{N}{\bf x}_i^*{\boldsymbol{\Sigma}_0^{-1}}{\bf x}_i}\right\|^{\frac{1}{2}}+\left\|\boldsymbol{\epsilon}_n(\rho)\right\|
\end{align*}
or equivalently:
\begin{align*}
&\|\hat{\bf C}_N(\rho)-{\boldsymbol{\Sigma}_0}\|\leq d_{\rm max}(\rho) \left\|\hat{\bf C}_N-\rho{\bf I}_N\right\|^{\frac{1}{2}}\left\|\boldsymbol{\Sigma}_0-\rho{\bf I}_N-\boldsymbol{\epsilon}_n\right\|^{\frac{1}{2}}\\
& + \left\|\boldsymbol{\epsilon}_n(\rho)\right\|.
\end{align*}
Since $d_{\rm max}(\rho)\asto 0$, to conclude, we need to check that the spectral norm of $\hat{\bf C}_N$ is almost surely bounded. The proof is almost the same as that proposed in Lemma \ref{lemma:bounded_spectral} to control the spectral norm of $\boldsymbol{\Sigma}_0$   with the slight difference that the expectation operator is replaced by the empirical average, and using additionally  the fact that $\frac{1}{n}\sum_{i=1}^n \frac{{\bf w}_i{\bf w}_i^*}{{\bf w}_i^*{\bf w}_i}\asto \frac{1}{N}{\bf I}_N$. Details are thus omitted.  
\section{Proof of Lemma \ref{lemma:di}}
\label{app:di}
The proof of Lemma \ref{lemma:di} is based on the same technique {as} in \cite{provost-94}. Using the relation $\frac{1}{\alpha}=\int_0^{+\infty} e^{-\alpha t} dt$, we write $\mathbb{E}\left[\frac{|w_i|^2}{{\bf w}^*{\bf D}{\bf w}}\right]$ as:
\begin{align*}
&\mathbb{E}\left[\frac{|w_i|^2}{{\bf w}^*{\bf D}{\bf w}}\right]=\mathbb{E}\left[|w_i|^2\int_0^{+\infty}e^{-t\left(d_i |w_i|^2+\sum_{j=1,j\neq i}^N  |w_j|^2 d_j\right)}\right] \\
&=\int_0^{+\infty}\int_0^{+\infty}\frac{1}{2^N}e^{-t d_i u} u \exp({-u/2})\int_0^{+\infty}\cdots\int_0^{+\infty} \\
&\times \exp\left({-t\displaystyle{\sum_{j=1,j\neq i} u_j d_j}}\right)\prod_{j=1,j\neq i}^Ne^{-u_j/2}du_1\cdots du_{N-1} du dt\\
&=\int_0^{\infty} \frac{1}{2^N}\frac{1}{(\frac{1}{2}+td_i)}  \prod_{j=1}^N \frac{1}{\frac{1}{2}+td_j}dt.
\end{align*}
Conducting the change of variable $t=\frac{1}{v}-1$, we eventually obtain:
$$
\mathbb{E}\left[\frac{|w_i|^2}{{\bf w}^*{\bf D}{\bf w}}\right]=\int_0^1 \frac{1}{2^N}\frac{v^{N-1}}{d_i\prod_{j=1}^N d_j (1-v\frac{d_i-\frac{1}{2}}{d_i})}\prod_{j=1}^N \frac{1}{1-v\frac{d_j-\frac{1}{2}}{d_j}}dv.
$$
We finally end the proof by using the integral representation of the Lauricella's type $D$ hypergeometric function.
\section{Proof of Lemma \ref{lemma:beta}}
Again the proof of the results in Lemma \ref{lemma:beta} {follows} the same lines as in {Appendix} \ref{app:di}. We will only detail the derivations for the expressions of $\beta_{i,i},i=1,\cdots,N$. The same kind of calculations can be used to derive that of $\beta_{i,j}, i\neq j$. Using the relation $\frac{1}{\alpha^2}=\int_0^{\infty}te^{-\alpha t}dt$, we write $\beta_{i,i}=\mathbb{E}\left[\frac{|w_i|^4}{({\bf w}^*{\bf D}{\bf w})^2}\right]$ as:
\begin{align*}
\beta_{i,i}& = \mathbb{E}\left[|w_i|^4\int_0^{\infty} te^{-t|w|_i^2+\sum_{j=1,j\neq i}|w_j|^2d_j}\right]\\
&=\int_0^{\infty}\int_0^{\infty}\frac{t}{2^N}u^2e^{-td_iu}u\exp(-u/2)\int_0^{\infty}\cdots \int_0^{\infty} \\
&\times \exp\left(-t\sum_{j=1,j\neq i}^N u_j d_j\right)\prod_{j=1,j\neq i}^N e^{-u_j/2}du_1\cdots du_{N-1}du dt\\
&=\frac{1}{2^{N-1}}\int_0^{\infty}\frac{t}{\left(\frac{1}{2}+td_i\right)^2} \prod_{k=1}^N \frac{1}{\frac{1}{2}+td_k}dt.
\end{align*}
Conducting the change of variable $t=\frac{1}{v}-1$, we obtain:
\begin{align*}
\beta_{i,i}&=\frac{1}{2^{N-1}}\int_0^1\frac{(1-v)v^{N-1}dv}{d_i^2\prod_{k=1}^N d_k\left(1-\frac{v(d_i-\frac{1}{2})}{d_i}\right)^2\prod_{k=1}^N(\frac{v(\frac{1}{2}-d_k)}{d_k}+1)}.
\end{align*}
\section{Proof of Theorem \ref{th:clt}}
Our approach is based on a perturbation analysis of ${\rm vec}(\hat{\bf C}_N(\rho))$ in the vicinity of the asymptotic limit ${\boldsymbol{\Sigma}_0}$ coupled with the use of the Slutsky Theorem \cite{vaart} which allows us to discard terms converging to zero in probability.

Set $\boldsymbol{\Delta}=\boldsymbol{\Sigma}_0^{-\frac{1}{2}}\left(\hat{\bf C}_N(\rho)-\boldsymbol{\Sigma}_0\right)\boldsymbol{\Sigma}_0^{-\frac{1}{2}}$. 
Then,
$$
\boldsymbol{\Delta}=\frac{N(1-\rho)}{n}\sum_{i=1}^n \frac{\boldsymbol{\Sigma}_0^{-\frac{1}{2}}{\bf x}_i{\bf x}_i^*\boldsymbol{\Sigma}_0^{-\frac{1}{2}}}{{\bf x}_i^*\hat{\bf C}_N^{-1}(\rho){\bf x}_i} +\rho\boldsymbol{\Sigma}_0^{-1}-{\bf I}_N.
$$
Writing $\hat{\bf C}_N^{-1}$ as:
\begin{align*}
\hat{\bf C}_N^{-1}&=\left(\hat{\bf C}_N-\boldsymbol{\Sigma}_0+\boldsymbol{\Sigma}_0\right)^{-1}\\
&=\boldsymbol{\Sigma}_0^{-\frac{1}{2}}\left({\bf I}_N+\boldsymbol{\Delta}\right)^{-1}\boldsymbol{\Sigma}_0^{-\frac{1}{2}}\\
&=\boldsymbol{\Sigma}_0^{-1}-\boldsymbol{\Sigma}_0^{-\frac{1}{2}}\boldsymbol{\Delta}\boldsymbol{\Sigma}_0^{-\frac{1}{2}}+o_p(\|\Delta\|)
\end{align*}
we obtain:
\begin{align*}
\boldsymbol{\Delta}&=\frac{N(1-\rho)}{n}\sum_{i=1}^n \frac{\boldsymbol{\Sigma}_0^{-\frac{1}{2}}{\bf x}_i{\bf x}_i^*\boldsymbol{\Sigma}_0^{-\frac{1}{2}}}{{\bf x}_i^*{\boldsymbol{\Sigma}_0^{-1}}{\bf x}_i-{\bf x}_i^*\boldsymbol{\Sigma}_0^{-\frac{1}{2}}\boldsymbol{\Delta}\boldsymbol{\Sigma}_0^{-\frac{1}{2}}{\bf x}_i+o_p(\|\boldsymbol{\Delta}\|)} \\
&+\rho\boldsymbol{\Sigma}_0^{-1}-{\bf I}_N.
\end{align*}
From \cite[Lemma 2.12]{vaart}, $\boldsymbol{\Delta}$ writes finally as:
\begin{align*}
\boldsymbol{\Delta}&=\frac{N(1-\rho)}{n}\sum_{i=1}^n\frac{\boldsymbol{\Sigma}_0^{-\frac{1}{2}}{\bf x}_i{\bf x}_i^*\boldsymbol{\Sigma}_0^{-\frac{1}{2}}}{{\bf x}_i^*\boldsymbol{\Sigma}_0^{-1}{\bf x}_i}\left(1+\frac{{\bf x}_i^*\boldsymbol{\Sigma}_0^{-\frac{1}{2}}\boldsymbol{\Delta}\boldsymbol{\Sigma}_0^{-\frac{1}{2}}{\bf x}_i}{{\bf x}_i^*\boldsymbol{\Sigma}_0^{-1}{\bf x}_i}\right)\\
&+\rho\boldsymbol{\Sigma}_0^{-1}-{\bf I}_N+o_p(\|\boldsymbol{\Delta}\|)\\
&=\boldsymbol{\Sigma}_0^{-\frac{1}{2}}\tilde{\boldsymbol{\Sigma}}\boldsymbol{\Sigma}_0^{-\frac{1}{2}}-{\bf I}_N\\
&+\frac{N(1-\rho)}{n}\sum_{i=1}^n \frac{\boldsymbol{\Sigma}_0^{-\frac{1}{2}}{\bf x}_i{\bf x}_i^*\boldsymbol{\Sigma}_0^{-\frac{1}{2}}{\bf x}_i^*\boldsymbol{\Sigma}_0^{-\frac{1}{2}}\boldsymbol{\Delta}\boldsymbol{\Sigma}_0^{-\frac{1}{2}}{\bf x}_i}{\left({\bf x}_i^*\boldsymbol{\Sigma}_0^{-1}{\bf x}_i\right)^2}+o_p(\|\boldsymbol{\Delta}\|)\\
&=\boldsymbol{\Sigma}_0^{-\frac{1}{2}}\tilde{\boldsymbol{\Sigma}}\boldsymbol{\Sigma}_0^{-\frac{1}{2}}-{\bf I}_N\\
&+\frac{N(1-\rho)}{n}\sum_{i=1}^n \frac{\boldsymbol{\Sigma}_0^{-\frac{1}{2}}{\bf x}_i{\bf x}_i^*\boldsymbol{\Sigma}_0^{-\frac{1}{2}}\left({\bf x}_i^{T}(\boldsymbol{\Sigma}_0^{-\frac{1}{2}})^{T}\otimes {\bf x}_i^*\boldsymbol{\Sigma}_0^{-\frac{1}{2}}\right) {\rm vec}(\boldsymbol{\Delta})}{\left({\bf x}_i^*\boldsymbol{\Sigma}_0^{-1}{\bf x}_i\right)^2}\\
&+o_p(\|\boldsymbol{\Delta}\|).
\end{align*}
Let  ${\bf F}$ be the $N^2\times N^2$ matrix given by:
$$
{\bf F}=\frac{N(1-\rho)}{n}\sum_{i=1}^n \frac{{\rm vec}\left(\boldsymbol{\Sigma}_0^{-\frac{1}{2}}{\bf x}_i{\bf x}_i^*\boldsymbol{\Sigma}_0^{-\frac{1}{2}}\right)\left({\bf x}_i^{T}(\boldsymbol{\Sigma}_0^{-\frac{1}{2}})^{T}\otimes {\bf x}_i^*\boldsymbol{\Sigma}_0^{-\frac{1}{2}}\right)}{\left({\bf x}_i^*\boldsymbol{\Sigma}_0^{-1}{\bf x}_i\right)^2}.
$$
Then, ${\rm vec}({\boldsymbol{\Delta}})$ satisfies the following system of equations: 
\begin{align}
{\rm vec}(\boldsymbol{\Delta})&={\rm vec}\left(\boldsymbol{\Sigma}_0^{-\frac{1}{2}}\tilde{\boldsymbol{\Sigma}}\boldsymbol{\Sigma}_0^{-\frac{1}{2}}-{\bf I}_N\right)+\mathbb{E}\left({\bf F}\right){\rm vec}(\boldsymbol{\Delta})\nonumber\\
&+\left({\bf F}-\mathbb{E}({\bf F})\right)\boldsymbol{\delta}+o_p(\|\boldsymbol{\delta}\|).
\label{eq:delta}
\end{align}
Given that the two last terms in the right-hand side of \eqref{eq:delta} converges to zero at a rate faster than $\frac{1}{\sqrt{n}}$, we have:
\begin{align}
&\sqrt{n}{\rm vec}(\boldsymbol{\Delta})=\sqrt{n}\left(\left(\boldsymbol{\Sigma}_0^{-\frac{1}{2}}\right)^{\mbox{\tiny T}}\otimes\boldsymbol{\Sigma}_0^{-\frac{1}{2}} \right)\tilde{\boldsymbol{\delta}}+\sqrt{n}\mathbb{E}({\bf F}){\rm vec}(\boldsymbol{\Delta})\nonumber\\
&+o_p(1).
\label{eq:delta_bis}
\end{align}
It remains thus to compute $\mathbb{E}({\bf F})$ and to check that its spectral norm is less than $1$. We will start by controlling the spectral norm of $\mathbb{E}({\bf F})$. Recall that $\mathbb{E}({\bf F})$ is given by:
\begin{align*}
&\mathbb{E}({\bf F})=N(1-\rho)\\
&\times\mathbb{E}\left[\frac{{\rm vec}\left(\boldsymbol{\Sigma}_0^{-\frac{1}{2}}{\bf x}{\bf x}^*\boldsymbol{\Sigma}_0^{-\frac{1}{2}}\right)\left({\bf x}^{\mbox{\tiny T}}(\boldsymbol{\Sigma}_0^{-\frac{1}{2}})^{\mbox{\tiny T}}\otimes {\bf x}^*\boldsymbol{\Sigma}_0^{-\frac{1}{2}}\right)}{({\bf x}^*\boldsymbol{\Sigma}_0^{-1}{\bf x})^2}\right]\\
&=N(1-\rho)\mathbb{E}\left[\frac{\left((\boldsymbol{\Sigma}_0^{-\frac{1}{2}})^{\mbox{\tiny T}}\overline{\bf x}\otimes \boldsymbol{\Sigma}_0^{-\frac{1}{2}}{\bf x}\right)\left({\bf x}^{\mbox{\tiny T}}(\boldsymbol{\Sigma}_0^{-\frac{1}{2}})^{\mbox{\tiny T}}\otimes {\bf x}^*\boldsymbol{\Sigma}_0^{-\frac{1}{2}}\right)}{\left({\bf x}^*\boldsymbol{\Sigma}_0^{-1}{\bf x}\right)^2}\right]\\
&=N(1-\rho)\mathbb{E}\left[\frac{\left((\boldsymbol{\Sigma}_0^{-\frac{1}{2}})^{\mbox{\tiny T}}\overline{\bf x}{\bf x}^{\mbox{\tiny T}}(\boldsymbol{\Sigma}_0^{-\frac{1}{2}})^{\mbox{\tiny T}}\right)\otimes\left(\boldsymbol{\Sigma}_0^{-\frac{1}{2}}{\bf x}{\bf x}^*\boldsymbol{\Sigma}_0^{-\frac{1}{2}}\right)}{\left({\bf x}^*\boldsymbol{\Sigma}_0^{-1}{\bf x}\right)^2}\right].
\end{align*}
It can be easily noticed that:
$\frac{(\boldsymbol{\Sigma}_0^{-\frac{1}{2}})^{\mbox{\tiny T}}\overline{\bf x}{\bf x}^{\mbox{\tiny T}}(\boldsymbol{\Sigma}_0^{-\frac{1}{2}})^{\mbox{\tiny T}}}{{\bf x}^*\boldsymbol{\Sigma}_0^{-1}{\bf x}}\preceq {\bf I}_N$. Therefore,
\begin{align*}
\mathbb{E}({\bf F}) &\preceq N(1-\rho)\mathbb{\bf I}_N\otimes \mathbb{E}\left[\frac{\boldsymbol{\Sigma}_0^{-\frac{1}{2}}{\bf x}{\bf x}^*\boldsymbol{\Sigma}_0^{-\frac{1}{2}}}{{\bf x}^*\boldsymbol{\Sigma}_0^{-1}{\bf x}}\right]\\
&=\mathbb{\bf I}_N\otimes\left({\bf I}_N-\rho\boldsymbol{\Sigma}_0^{-1}\right)
\end{align*}
thus implying
$$
\left\|\mathbb{E}({\bf F})\right\|\leq \left\|{\bf I}_N-\rho\boldsymbol{\Sigma}_0^{-1}\right\| <1.
$$
We will now provide a closed-form expression for $\mathbb{E}({\bf F})$. To this end, we will use the eigenvalue decomposition of $\boldsymbol{\Sigma}_0^{-\frac{1}{2}}\boldsymbol{\Sigma}_N^{\frac{1}{2}}={\bf U}{\bf D}^{\frac{1}{2}}{\bf U}^*$.
Then, letting $\tilde{\bf w}={\bf U}^*{\bf w}$ with ${\bf w}=\boldsymbol{\Sigma}_N^{-\frac{1}{2}}{\bf x}$, we obtain:
$$
\mathbb{E}({\bf F})=\mathbb{E}\left[\frac{N(1-\rho)\overline{\bf U}{\bf D}^{\frac{1}{2}}\overline{(\tilde{\bf w})}\tilde{\bf w}^{\mbox{\tiny T}}{\bf D}^{\frac{1}{2}}{\bf U}^{\mbox{\tiny T}}\otimes {\bf U}\boldsymbol{\bf D}^{\frac{1}{2}}\tilde{\bf w}\tilde{\bf w}^*{\bf D}^{\frac{1}{2}}{\bf U}^*}{\left(\tilde{\bf w}^*{\bf D}\tilde{\bf w}\right)^2}\right].
$$
Therefore,
\begin{align*}
&\left({\bf D}^{-\frac{1}{2}}{\bf U}^{\mbox{\tiny T}}\otimes {\bf D}^{-\frac{1}{2}}{\bf U}^*\right)\mathbb{E}({{\bf F}}) \left(\overline{\bf U}{\bf D}^{-\frac{1}{2}}\otimes {\bf U}{\bf D}^{-\frac{1}{2}}\right)\\
&=N(1-\rho)\mathbb{E}\left[\frac{\overline{(\tilde{\bf w})}\tilde{\bf w}^{\mbox{\tiny T}}\otimes \tilde{\bf w}\tilde{\bf w}^*}{\left(\tilde{\bf w}^*{\bf D}\tilde{\bf w}\right)^2}\right]\\
&=N(1-\rho)\mathbb{E}\left[\frac{\left(\overline{(\tilde{\bf w})}\otimes \tilde{\bf w}\right)(\tilde{\bf w}\otimes \tilde{\bf w}^*)}{\left(\tilde{\bf w}^*{\bf D}\tilde{\bf w}\right)^2}\right]\\
&=N(1-\rho)\mathbb{E}\left[\frac{{\rm vec}(\tilde{\bf w}\tilde{\bf w}^*)\left({\rm vec}(\tilde{\bf w}\tilde{\bf w}^*)\right)^*}{\left(\tilde{\bf w}^*{\bf D}\tilde{\bf w}\right)^2}\right]\\
&=N(1-\rho)\tilde{\bf B}({\bf D}),
\end{align*}
where $\tilde{\bf B}({\bf D})$ is provided by Lemma \ref{lemma:di}. A closed-form expression for $\tilde{\bf F}\triangleq\mathbb{E}({\bf F})$ is thus given by:
$$
\tilde{\bf F}= N(1-\rho)\left(\overline{\bf U}{\bf D}^{\frac{1}{2}}\otimes {\bf U}{\bf D}^{\frac{1}{2}}\right)\tilde{\bf B}({\bf D})\left({\bf D}^{\frac{1}{2}}{\bf U}^{\mbox{\tiny T}}\otimes {\bf D}^{\frac{1}{2}}{\bf U}^*\right).
$$ 
The linear system of equations in \eqref{eq:delta_bis} thus becomes:
$$
\sqrt{n}{\rm vec}(\boldsymbol{\Delta})=\sqrt{n}({\bf I}_N-\tilde{\bf F})^{-1}\left(\left(\boldsymbol{\Sigma}_0^{-\frac{1}{2}}\right)^{\mbox{\tiny T}}\otimes \boldsymbol{\Sigma}_0^{-\frac{1}{2}}\right)\tilde{\boldsymbol{\delta}}+o_p(1).
$$
Writing ${\rm vec}(\boldsymbol{\Delta})=\left(\left(\boldsymbol{\Sigma}_0^{-\frac{1}{2}}\right)^{\mbox{\tiny T}}\otimes \boldsymbol{\Sigma}_0^{-\frac{1}{2}}\right)\boldsymbol{\delta}$, 
 we finally obtain:
\begin{align*}
\sqrt{n}\boldsymbol{\delta}&=\left(\left(\boldsymbol{\Sigma}_0^{\frac{1}{2}}\right)^{\mbox{\tiny T}}\otimes \boldsymbol{\Sigma}_0^{\frac{1}{2}}\right)({\bf I}_{N^2}-\tilde{\bf F})^{-1}\left(\left(\boldsymbol{\Sigma}_0^{-\frac{1}{2}}\right)^{\mbox{\tiny T}}\otimes \boldsymbol{\Sigma}_0^{-\frac{1}{2}}\right)\sqrt{n}\tilde{\boldsymbol{\delta}}\\
&+o_p(1).
\end{align*}
Thus, $\sqrt{n}\boldsymbol{\delta}$ behaves as a zero-mean Gaussian distributed vector with covariance:
\begin{align*}
{\bf M}_1&=\left(\left(\boldsymbol{\Sigma}_0^{\frac{1}{2}}\right)^{\mbox{\tiny T}}\otimes \boldsymbol{\Sigma}_0^{\frac{1}{2}}\right)({\bf I}_{N^2}-\tilde{\bf F})^{-1}\left(\left(\boldsymbol{\Sigma}_0^{-\frac{1}{2}}\right)^{\mbox{\tiny T}}\otimes \boldsymbol{\Sigma}_0^{-\frac{1}{2}}\right)\tilde{\bf M}_1\\
&\times\left(\left(\boldsymbol{\Sigma}_0^{-\frac{1}{2}}\right)^{\mbox{\tiny T}}\otimes \boldsymbol{\Sigma}_0^{-\frac{1}{2}}\right)({\bf I}_{N^2}-\tilde{\bf F})^{-1}\left(\left(\boldsymbol{\Sigma}_0^{\frac{1}{2}}\right)^{\mbox{\tiny T}}\otimes \boldsymbol{\Sigma}_0^{\frac{1}{2}}\right)
\end{align*}
and pseudo-covariance:
\begin{align*}
{\bf M}_2&=\left(\left(\boldsymbol{\Sigma}_0^{\frac{1}{2}}\right)^{\mbox{\tiny T}}\otimes \boldsymbol{\Sigma}_0^{\frac{1}{2}}\right)({\bf I}_{N^2}-\tilde{\bf F})^{-1}\left(\left(\boldsymbol{\Sigma}_0^{-\frac{1}{2}}\right)^{\mbox{\tiny T}}\otimes \boldsymbol{\Sigma}_0^{-\frac{1}{2}}\right)\tilde{\bf M}_2\\
&\times\left(\boldsymbol{\Sigma}_0^{-\frac{1}{2}}\otimes \left(\boldsymbol{\Sigma}_0^{-\frac{1}{2}}\right)^{\mbox{\tiny T}}\right)({\bf I}_{N^2}-\tilde{\bf F}^{\mbox{\tiny T}})^{-1}\left(\boldsymbol{\Sigma}_0^{\frac{1}{2}}\otimes \left(\boldsymbol{\Sigma}_0^{\frac{1}{2}}\right)^{\mbox{\tiny T}}\right)
\end{align*}
This completes the proof.
\label{app:clt}
\bibliographystyle{IEEEbib}
\bibliography{IEEEabrv,IEEEconf,./tutorial_RMT.bib}

\end{document}